\documentclass[11pt,fleqn]{llncs}

\usepackage{syntonly}
\usepackage{latexsym}
\usepackage{amssymb}
\usepackage{epsfig}

\newcommand{\msqrt}{\mbox{\rm msqrt}}
\newcommand{\mlog}{\mbox{\rm mlog}}
\newcommand{\msqrti}{\mbox{\footnotesize\rm msqrt}}
\newcommand{\mlogi}{\mbox{\footnotesize\rm mlog}}
\newcommand{\imax}{\mbox{\footnotesize\rm max}}

\newcommand{\ord}{\mbox{$o$}}
\newcommand{\Ord}{\mbox{$O$}}
\newcommand{\DT}{\mbox{\rm DTIME}}
\newcommand{\DTi}{\mbox{\footnotesize\rm DTIME}}
\newcommand{\NT}{\mbox{\rm NTIME}}
\newcommand{\AT}{\mbox{\rm ATIME}}
\newcommand{\XT}{\mbox{\rm XTIME}}

\newcommand{\X}{\mbox{$\mathbb{X}$}}

\newcommand{\Nd}{\mbox{\rm N}}
\newcommand{\code}{\mbox{\rm code}}

\newcommand{\Blin}{\mbox{$B_{\mbox{\footnotesize\rm lin}}$}}

\newcommand{\Bpol}{\mbox{$B_{\mbox{\footnotesize\rm pol}}$}}

\newcommand{\Bqpol}{\mbox{$B_{\mbox{\footnotesize\rm qpol}}$}}

\newcommand{\Bqlin}{\mbox{$B_{\mbox{\footnotesize\rm qlin}}$}}

\newcommand{\Bqmlin}[1]{\mbox{$B_{\mbox{\footnotesize\rm q$^#1$lin}}$}}

\newcommand{\Bhex}{\mbox{$B_{\mbox{\footnotesize\rm hex}}$}}

\newcommand{\Blogs}{\mbox{$B_{\mbox{\footnotesize\rm $\log^{*}$}}$}}

\newcommand{\ite}[1]{\mbox{$^{\langle#1\rangle}$}}  

\newcommand{\leqa}{\mbox{$\leq_{\mbox{\footnotesize\rm a}}$}}
\newcommand{\eqa}{\mbox{$\equiv_{\mbox{\footnotesize\rm a}}$}}
\newcommand{\leqo}{\mbox{$\leq_{\mbox{\footnotesize\rm ord}}$}}
\newcommand{\leo}{\mbox{$\ll_{\mbox{\footnotesize\rm ord}}$}}  
\newcommand{\nleqo}{\mbox{$\not\leq_{\mbox{\footnotesize\rm ord}}$}}
\newcommand{\leqi}{\mbox{$\leq_{\mbox{\footnotesize\rm it}}$}}
\newcommand{\lei}{\mbox{$<_{\mbox{\footnotesize\rm it}}$}}
\newcommand{\eqi}{\mbox{$\equiv_{\mbox{\footnotesize\rm it}}$}}
\newcommand{\leqpow}{\mbox{$\leq_{\mbox{\footnotesize\rm pow}}$}}
\newcommand{\lepow}{\mbox{$<_{\mbox{\footnotesize\rm pow}}$}}
\newcommand{\llpow}{\mbox{$\ll_{\mbox{\footnotesize\rm pow}}$}}
\newcommand{\leqx}{\mbox{$\leq_{\mbox{\footnotesize\rm lex}}$}}
\newcommand{\wlex}{\mbox{$\sqsubseteq_{\mbox{\footnotesize\rm lex}}$}}
\newcommand{\eqpow}{\mbox{$\equiv_{\mbox{\footnotesize\rm pow}}$}}

\newcommand{\emi}{\mbox{$\sqsubseteq_{\mbox{\footnotesize\rm it}}$}}
\newcommand{\isi}{\mbox{$\cong_{\mbox{\footnotesize\rm it}}$}}
\newcommand{\incl}{\mbox{$\subseteq_{\mbox{\footnotesize\rm low}}$}}

\newcommand{\It}{\mbox{\rm It}}
\newcommand{\Pow}{\mbox{\rm Pow}}
\newcommand{\Cl}{\mbox{\rm Clos}}
\newcommand{\Reg}{\mbox{$\cal R\!$\it eg}}
\newcommand{\Clit}{\mbox{$\cal C\!$\it lit}}
\newcommand{\Csets}{\mbox{$\cal C\!$\it sets}}
\newcommand{\Cuts}{\mbox{$\cal C\!$\it uts}}
\newcommand{\val}{\mbox{\it val}}

\newcommand{\ov}[1]{\mbox{$\overline{#1}$}}

\newcommand{\Sa}[2]{\mbox{$\Sigma_{#1}(#2)$}}
\newcommand{\Pa}[2]{\mbox{$\Pi_{#1}(#2)$}}
\newcommand{\Salt}[2]{\mbox{$\Sigma_{#1}^{\mbox{\footnotesize\rm alt}}(#2)$}}
\newcommand{\Palt}[2]{\mbox{$\Pi_{#1}^{\mbox{\footnotesize\rm alt}}(#2)$}}
\newcommand{\Sai}[2]{\mbox{\footnotesize $\Sigma_{#1}(#2)$}}

\newcommand{\OH}{\mbox{\rm OH}}

\newcommand{\NP}{\mbox{\rm NP}}
\renewcommand{\P}{\mbox{\rm P}}
\newcommand{\Lin}{\mbox{\rm LIN}}
\newcommand{\NLin}{\mbox{\rm NLIN}}

\newcommand{\co}{\mbox{\rm co}}

\newcommand{\N}{\mbox{$\mathbb{N}$}}                        
\newcommand{\Ni}{\mbox{\footnotesize $\mathbb{N}$}}         
\newcommand{\Nplus}{\mbox{$\mathbb{N}_+\,$}}
\newcommand{\Nplusi}{\mbox{\footnotesize $\mathbb{N}_+\,$}}                

\newcommand{\B}{\mbox{$\mathbb{B}$}}
\newcommand{\D}{\mbox{$\mathbb{D}$}}
\newcommand{\Q}{\mbox{$\mathbb{Q}$}}                        
\newcommand{\R}{\mbox{$\mathbb{R}$}}                        

\newcommand{\fct}[3]{\mbox{$#1 :\, #2\, \longrightarrow \, #3$}}

\newcommand{\iid}{\mbox{\footnotesize\rm id}}

\newcommand{\ran}{\mbox{\rm ran}}

\newcommand{\en}{\enspace}

\newcommand{\bea}{\begin{eqnarray*}}
\newcommand{\eea}{\end{eqnarray*}}

\begin{document}












\vspace*{0.3cm}

\begin{center}\Large\bf  On Regular Sets of Time Bounds
 \\  
 and Determinism versus Nondeterminism
  \vspace*{0.5cm}

\normalsize\rm Armin Hemmerling
 \vspace*{0.12cm}

\small Ernst-Moritz-Arndt--Universit\"at Greifswald,
           Institut f\"ur Mathematik und Informatik\\
 Walther-Rathenau-Str. 47,
           D--17487 Greifswald, Germany\\
{\small\tt hemmerli@uni-greifswald.de}
 \vspace*{0.1cm}

\normalsize
September 2013
\end{center}
 \vspace*{1.cm}

\hspace*{0.25cm}
\begin{minipage}{14.cm}
\small
\noindent
 {\bf Abstract.} \en
This paper illustrates the richness of the concept of regular sets of time bounds and demonstrates its application to problems
of computational complexity. 
There is a universe of bounds whose regular subsets allow to represent several time complexity classes of common interest
and are linearly ordered with respect to the confinality relation which implies the inclusion between the corresponding
complexity classes.
By means of classical results of complexity theory, the separation of determinism from nondeterminism is possible
for a variety of sets of bounds below $n\cdot\log^*(n)$.
The system of all regular bound sets ordered by confinality allows the order-isomorphic embedding of, e.g., the ordered set
of real numbers or the Cantor discontinuum.
   \\[1.2ex]
  \end{minipage}
\vspace*{1.cm}

\section{Introduction and overview}

Complexity classes usually consist of all languages whose decision problems are solvable by means of a fixed underlying model of computation with a bounded amount of certain resources. Time complexity classes defined by deterministic, nondeterministic or alternating Turing machines occupy a considerable part of interest in computational complexity theory. They can be characterized by single time bounds as well as by sets of time bounds. In \cite{He1},  we introduced the concept of  regular sets of time bounds.
It guarantees that the related complexity classes own a certain robustness, in particular they are closed with respect to some standard techniques of Turing machine programming. Moreover, there are special languages which are complete in all such nondeterministic complexity classes, with respect  to related versions of m-reducibility. As an immediate consequence, the equality of the deterministic and the nondeterministic class, both of them defined by the same regular set of bounds,
 is upwards hereditary with respect to inclusion between the  sets of bounds.
Similar relationships hold for collapses of oracle hierarchies defined analogously to the linear time hierarchy or the polynomial time hierarchy, provided that the underlying set of bounds fulfills the condition of o-regularity.

Besides the sets of linear bounds and of polynomial bounds, respectively, there are some further o-regular bound sets yielding well-known complexity classes, see \cite{He1} and Section 2 of this paper. Nevertheless, so far it has been open how large the families of regular or even o-regular sets are and how useful or interesting the related results should be regarded. The present paper tries to contribute some insight concerning such questions.

After the basic notions and notations have been introduced in Section 2, Section 3 deals with regular subsets of tame universes
of time bounds. It turns out that regularity is closely connected with the sets of iterations of bounds.
Section 4 gives some further details and deals with properties of universes of bounds which enable us to represent a variety of commonly used complexity classes.
In Section 5, a universe of special bounds between linear and polynomial time is constructed. It is linearly ordered by $\leqi$,
this relation is closely related to the inclusion of the corresponding complexity classes.
Section 6 describes an analog construction of a universe of bounds between the polynomials and the quasipolynomials.
In Section 7, it is shown that the union of the universes from Sections 5 and 6 yields an order structure of regular sets
which is order-isomorphic to the Cantor discontinuum enriched by just one element greater than all the others.
In Section 8, the regular sets of bounds from Sections 5 and 6 are shown to be o-regular. Thus, oracle hierarchies and,
equivalently, alternating hierarchies can be defined over these sets. Applying a classical result of complexity theory,
we obtain the separation of deterministic from nondeterministic complexity classes defined by certain sets of bounds
below $n\cdot\log^*(n)$.
In the concluding Section 9, we explain that questions concerning determinism versus nondeterminism or
the problem of collapses of oracle hierarchies determine cuts in linearly ordered universes of bounds.

Even if the present paper does not solve any crucial problem, we are convinced that this introduction into the world
of regular sets of bounds and their complexity classes contributes some interesting new features of computational complexity.


\section{Time bounds and regular sets}

We start with recalling some basic notions and denotations. Most of them were already used in \cite{He1,He2},
and their definitions are essentially taken from there.

By a  {\it bound function}\/ (also {\it time bound}\/ or briefly {\it bound}\/), we always mean a function over the natural numbers, $\fct{\beta}{\N}{\N}$, satisfying $n\leq \beta(n)\leq \beta(n+1)$ for all $n\in\N$. The related {\it (time) complexity classes}\/ are
\bea
\XT(\beta) & = & \{L: \; L\subseteq \X^* \mbox{ and there is an $\mbox{\rm X}$TM $\frak{M}$ } \\
 & & \hspace*{2.2cm}\mbox{accepting $L$ with a time complexity $t_{\frak{M}}\in \Ord(\beta)$} \},
\eea
where $\mbox{\rm X}\in\{\mbox{\rm D}, \mbox{\rm N}, \mbox{\rm A}\}$.
They consist of {\it languages}\/ $L$ over the two-letter alphabet $\X=\{0,1\}$.
Under an $\mbox{\rm X}$TM, we here understand a {\it deterministic}\/, {\it nondeterministic}\/ and {\it alternating}\/, respectively, {\it Turing machine}\/ with a special read-only input tape and arbitrarily many additional work tapes. Further details on standard definitions can be found in textbooks like \cite{DK,Pa,Re}.

For a set $B$ of bounds and $\mbox{\rm X}\in\{\mbox{\rm D}, \mbox{\rm N}, \mbox{\rm A}\}$,
the complexity classes are defined as
\bea
 \XT(B) & = & \bigcup\nolimits_{\beta\in B} \XT(\beta).
\eea
In particular, we obtain the well-known deterministic and nondeterministic {\it polynomial time}\/ classes
in this way:
\[  \P= \DT(\Bpol)\en \mbox{ and } \en \NP= \NT(\Bpol), \]
where
\bea
\Bpol&=& \{\beta: \:\beta(n)=  n^k \mbox{ with some } k\in\Nplus\}.
\eea
This set represents the polynomial bounds. Notice that $0\in\N$ and $\Nplus=\N \setminus\{0\}$.
The (deterministic and nondeterministic)  {\it linear time}\/ classes are
\[\Lin =\DT(\Blin) \en \mbox{ and }\en \NLin =\NT(\Blin), \]
where
\bea
\Blin=\{\beta: \: \beta(n)= c \cdot n \mbox{ with some } c\in\Nplus\}
\eea
denotes the set of linear bounds.
The linear time classes can even be characterized by a single bound, e.g., by the identical function, $\beta_{\iid}(n)=n\,$:
\[\Lin =\DT(\beta_{\iid})\,,\en \NLin =\NT(\beta_{\iid})  \en \mbox{ and }\en \AT(\Blin)=\AT(\beta_{\iid})    . \]
Another important set of bounds is
\[\Bqlin = \{\beta: \:\beta(n)=  n\cdot (\,\log(n)\,)^k \mbox{ with } k\in\Nplus\}\]
characterizing the {\it quasilinear}\/ time classes, where
$\log(n)=\lceil\log_2(n)\rceil$ for $n\geq 2$ and $\log(0)=\log(1)=1$.
Also, the modifications
\[\Bqmlin{m} = \{\beta: \:\beta(n)=  n\cdot (\log\ite{m}(n))^k \mbox{ with } k\in\Nplus\},
\]
are of some interest. Herein, $\log\ite{m}$ is the $m$-fold iteration of $\log$. Quite general, for $m\in\Nplus$, the $m$-fold iteration of a function $f$ will be denoted by $f\ite{m}$. It must not be confused with $f^m$ standing for the arithmetical power of a function $f$ mapping into a number domain, i.e., $f^m(n)= (\,f(n)\,)^m$.

Also, the bound set
\[\Blogs = \{\beta: \:\beta(n)=  n\cdot (\log^*(n))^k \mbox{ with } k\in\Nplus\},
\]
with the {\it iterated logarithm}\/, $\log^*(n)= \min\{m:\, \log\ite{m}(n)=1\}$, and the set of
{\it quasipolynomial}\/ (or superpolynomial) bounds,
\[\Bqpol = \{\beta: \:\beta(n)=  n^{(\log(n))^k} \mbox{ with } k\in\Nplus\},\]
yield complexity classes known from literature or derived from such ones.
Finally, the set of iterations of the exponential function,
$\beta_{\mbox{\footnotesize exp}}(n)= 2^n$, should be mentioned:
let
\[\Bhex = \{\beta_{\mbox{\footnotesize exp}}\ite{m}: \: m\in\Nplus\}.\]
We have $\beta_{\mbox{\footnotesize exp}}\ite{m}(n)= 2^{2^{\cdot^{\cdot^{\cdot^{2^n}}}}}$, with $m$-times $2\,$s.

For bound functions $\fct{\beta_1,\beta_2}{\N}{\N}$ (or any functions $\beta_1,\beta_2$ mapping into number domains),  let $\beta_1\leqa \beta_2$ mean that $\beta_1(x)\leq \beta_2(x)$ for almost all arguments $x$. Moreover, operations on numbers are naturally transferred to such functions,
e.g., $(\beta_1+\beta_2)(n)=\beta_1(n)+\beta_2(n)$ and $(c\cdot \beta)(n)=c\cdot\beta(n)$ for constants $c\in\N$.
As usual, functions $\beta$ will also be denoted by terms describing them; e.g., $\beta(n)\leqa\, n$ means that $\beta\leqa\beta_{\iid}$,
with the identical function $\beta_{\iid}(n)=n$ on the right-hand side.

Obviously, relation $\leqa$ is reflexive and transitive.
Moreover, we shall have to use the monotony with respect to addition, multiplication and composition of functions.

\begin{lemma}
If $\beta_1\leqa\beta_2$ and $\beta_1'\leqa\beta_2'$ for bounds $\beta_1,\beta_2,\beta_1'$and $\beta_2'$, then  we have $\beta_1+\beta_1'\leqa\beta_2+\beta_2'$, $\beta_1\cdot\beta_1'\leqa\beta_2\cdot\beta_2'$, and  $\beta_1\circ\beta_1'\leqa\beta_2\circ\beta_2'$.
\end{lemma}

\begin{proof}
Whereas the first two assertions hold for any number functions, the proof of the last one uses the special properties of time bounds. Let $\beta_1(n)\leq\beta_2(n)$ for all $n\geq c$ and $\beta_1'(n)\leq\beta_2'(n)$ for all $n\geq c'$,
where $c,c'\in \N$. Then $(\beta_1+\beta_2)(n)\leq (\beta_1'+\beta_2')(n)$ and
$(\beta_1\cdot\beta_2)(n)\leq (\beta_1'\cdot\beta_2')(n)$ for all $n\geq \max(c,c')$.
Moreover,
\[\beta_1\circ\beta_1'(n)=\beta_1(\beta_1'(n))\leq \beta_1(\beta_2'(n))\leq
 \beta_2(\beta_2'(n))=\beta_2\circ\beta_2'(n)\]
if $n\geq c'$ and $\beta_2'(n)\geq c$; this holds for all  $n\geq \max(c,c')$, too.
 \qed
 \end{proof}

In order to compute word functions (within certain complexity bounds), we employ  DTMs with an additional one-way write-only output tape on which the values of the functions have to be produced.
A bound function $\beta$ is called {\it time-constructible}\/ iff the word function
$\fct{f_{\beta}}{\X^*}{\X^*}$, defined by $f_{\beta}(w)= 1^{\beta(|w|)}$, is (deterministically) computable in time $\Ord(\beta)$. Equivalently, one could require that there is a deterministic multi-tape Turing machine which halts on any input $w\in \X^*$ after $\Ord(\beta(|w|))$ steps such that the last work tape carries the inscription $1^{\beta(|w|)}$ then.

More general, a bound $\beta$ is said to be {\it constructible in time} $\beta'$ if $f_{\beta}$ is computable in time $\Ord(\beta')$. 
In particular, {\it linear-time}\/ constructibility, i.e., constructibility in time $\Ord(n)$ will be employed later. Instead of  the word function $f_{\beta}$ which is related to the unary notation of numbers, the function $f_{\beta}'$, where $f_{\beta}'(w)$ is the binary encoding of $|w|$, could equivalently be used. For all these details on time-constructibility, the reader is referred to \cite{Ka}.
One easily shows the time-constructibility of several commonly used bound functions as well as the following
\begin{lemma}
If $\beta_1$ and $\beta_2$ are time-constructible bounds, then $\beta_1+\beta_2$, $\beta_1\cdot\beta_2$ and
$\beta_1\circ\beta_2$ are time-constructible, too.
\qed
\end{lemma}

\begin{definition}[Regularity]
A non-empty set of bounds, $B$, is called  regular iff the following properties hold:
\\[1ex]
\hspace*{1.5cm}
\begin{minipage}{14.cm}
\begin{itemize}
\item[ i)] for all $\beta\in B$, there is a time-constructible $\beta'\in B$ such that
$\beta\leqa \beta'$;
\item[ii)] for all $\beta,\beta'\in B$, there is a $\beta''\in B$ such that
$\beta + \beta'\circ \beta \leqa \beta''$.
\end{itemize}
\end{minipage}
\end{definition}

This notion introduced in \cite{He1} guarantees a certain robustness of the complexity classes related to a bound set $B$ such that some key techniques of complexity theory are applicable. Notice that condition ii) implies that, for any $\beta_1,\beta_2,\beta\in B$ and any constant $c\in\Nplus$, there are $\beta',\beta''\in B$
such that $\beta_1+\beta_2\leqa\beta'$ and $c\cdot\beta\leqa\beta''$.
It is easily shown that all the bound sets mentioned so far, i.e.,
$\Blin, \Bpol, \Bqlin, \Bqmlin{m}, \Blogs, \Bqpol$ and $\Bhex$, are regular.

Relation $\leqa$ is canonically transferred to sets of bounds:
let $B_1\leqa B_2$ mean that to any $\beta_1\in B_1$ there is a
$\beta_2\in B_2$ such that $\beta_1\leqa\beta_2$. Also this relation between sets
is reflexive and transitive. Obviously, $B_1\subseteq B_2$ implies that $B_1\leqa B_2$.
Moreover, one easily shows that $\Blin\leqa B$ for each regular set $B$, and
we have
\[\Blin \leqa \Blogs \leqa \,\cdots \,\leqa \Bqmlin{3} \leqa \Bqmlin{2}\leqa \Bqmlin{1} =\Bqlin
\leqa \Bpol\leqa\Bqpol\leqa\Bhex.\]

By $\eqa$, the equivalence relation related to $\leqa$ is denoted, i.e., $B_1\eqa B_2$ means that
both $B_1\leqa B_2$ and $B_2\leqa B_1$. We shall then say that $B_1$ and $B_2$ are {\it confinal}\/.

Clearly, $B_1\leqa B_2$ implies $\XT(B_1)\subseteq \XT(B_2)$ for each $\mbox{\rm X}\in\{\mbox{\rm D}, \mbox{\rm N}, \mbox{\rm A}\}$. Thus, confinal sets of bounds yield the same complexity classes.

With respect to inclusions between complexity classes, the {\it order of growth}\/ of bound functions is
still more important than relation $\leqa$.
We write $\beta_1\leqo\beta_2$ iff $\beta_1\in\Ord(\beta_2)$, i.e., there is a constant $c\in\Nplus$ such that $\beta_1\,\leqa \, c \cdot \beta_2$. By $\beta_1\,\leo \, \beta_2$, we denote that
$\beta_1\in\ord(\beta_2)$, this means
\[\lim_{n\to \infty \atop (n>0)} \,\mbox{\normalsize $\frac{\beta_1(n)}{\beta_2(n)}$}\, =\,\limsup_{n\to \infty \atop (n>0)} \,\mbox{\normalsize $\frac{\beta_1(n)}{\beta_2(n)}$}\,
=\,\liminf_{n\to \infty \atop (n>0)} \,\mbox{\normalsize $\frac{\beta_1(n)}{\beta_2(n)}$}\,
=\,\infty.\]
In the sequel, we shall simply write
$\,\lim_{n\to \infty} \,\mbox{\normalsize $\frac{\beta_1(n)}{\beta_2(n)}$}$ for such limits, and analogously for
limits superior and inferior, even if $\beta_2(0)=0$ is not excluded.

It is easily shown that relation $\leqo$ is reflexive and transitive. The following lemma concerns its monotony.
A bound function $\beta$ is called {\it superlinear}\/ iff $\beta_{\iid}\leo\beta$.

\begin{lemma}
From $\beta_1\leqo\beta_2$ it follows that $\beta_1\circ\beta\leqo\beta_2\circ\beta$ for all bounds $\beta$.
If $\beta$ is superlinear, then it holds $\,\beta'\leo\beta\circ\beta'$ for all bounds $\beta'$.
\end{lemma}

\begin{proof}
$\beta_1\leqo\beta_2$ means that $\limsup_{n\to \infty} \,\mbox{\normalsize $\frac{\beta_1(n)}{\beta_2(n)}$}<\infty$. Then $\limsup_{n\to \infty} \,\mbox{\normalsize $\frac{\beta_1\circ\beta(n)}{\beta_2\circ\beta(n)}$}\in \R$, too.
To show the second assertion, notice that always $m\leq\beta(m)$, hence $\beta'(n)\leq\beta\circ\beta'(n)$ for all $n\in\N$.
If $\beta'\leo\beta\circ\beta'$ would not hold, then we had
$\liminf_{n\to \infty} \,\mbox{\normalsize $\frac{\beta\circ\beta'(n)}{\beta'(n)}$}<\infty$,
hence $\liminf_{m\to \infty} \,\mbox{\normalsize $\frac{\beta(m)}{m}$}<\infty$, a contradiction to
$\lim_{m\to \infty} \,\mbox{\normalsize $\frac{\beta(m)}{m}$}=\infty$.
\qed
\end{proof}

It might be of interest that $\beta_1\leqo\beta_2$ always implies $\beta \circ\beta_1\leqo\beta\circ\beta_2$
iff $\beta$ is {\it subhomogeneous}\, i.e., to any number $c\in\Nplus$ there exists a number $\overline{c}$ such that 
$\beta(c\cdot n)\leq \overline{c}\cdot \beta(n)$ for all $n\in\Nplus$, see \cite{He2}. \\
Indeed, if $\beta$ is subhomogeneous and $\beta_1\leqo\beta_2$, i.e., $\beta_1\leqa\beta_2$ for some $c\in\Nplus$,
then it follows $\beta\circ\beta_1(n)\leq\beta(c\cdot\beta_2(n))\leq \overline{c}\cdot \beta\circ\beta_2(n)$ for almost all $n\in\N$.
Conversely, if  $\beta_1\leqo\beta_2$ always implies $\beta \circ\beta_1\leqo\beta\circ\beta_2$, take 
$\beta_1(n)=c\cdot n$ and $\beta_2(n)=n$. Then we have $\beta(c\cdot n)\leq c'\cdot n$ for almost all $n$ with a suitably chosen $c'\in\Nplus$, and it follows $\beta(c\cdot n)\leq \overline{c}\cdot n$ for all $n\in\Nplus$, with some $\overline{c}\in \Nplus$.

It was shown in \cite{He2} that subhomogeneity is a natural property satisfied by a variety of bounds up to the polynomial ones. On the other hand, as already known from \cite{Ka}, superpolynomial bounds cannot be subhomogeneous.

For sets of bounds, $B_1$ and $B_2$,
let $B_1\leqo B_2$ mean that to any $\beta_1\in B_1$ there is a
$\beta_2\in B_2$ with $\beta_1\leqo\beta_2$. It is well-known that this suffices to imply $\XT(B_1)\subseteq \XT(B_2)$ for
all prefixes
 $\mbox{\rm X}\in\{\mbox{\rm D}, \mbox{\rm N}, \mbox{\rm A}\}$.

\begin{lemma}
For regular sets of bounds, $B_1$ and $B_2$, it holds \en
$ B_1\leqa B_2 \en \mbox{ iff } \en B_1\leqo B_2. $
\end{lemma}

\begin{proof}
Indeed, from $B_1\leqa B_2$ it always follows that $B_1\leqo B_2$; here the regularity is not needed.
Conversely, let $B_1\leqo B_2$. Thus, for any $\beta_1\in B_1$ there is a $\beta_2\in B_2$ with $\beta_1\,\leqo\,\beta_2$, i.e., $\beta_1\leqa\, c\cdot\beta_2$ for some $c\in\Nplus$.
If $B_2$ is regular, there is a $\beta_2'\in B_2$ satisfying $c\cdot \beta_2 \leqa \beta_2'$, hence $\beta_1\leqa\beta_2'$.
\qed
\end{proof}

In the sequel, we shall simply write $B_1\leq B_2$ instead of $B_1\leqa B_2$ or $B_1\leqo B_2$, for regular sets
$B_1$ and $B_2$. Let $\,\equiv\,$ denote the related equivalence relation, i.e.,
$B_1\,\equiv\, B_2$  iff both $B_1\leq B_2$  and $B_2\leq B_1$. This means that
 $B_1$ and $B_2$ are confinal. Finally, $B_1 < B_2$ means that $B_1\leq B_2$ but not $B_2\leq B_1$, i.e.,  $B_1\,\not\equiv\, B_2$. One easily verifies that
\[\Blin < \Blogs < \,\cdots \,< \Bqmlin{3} < \Bqmlin{2}< \Bqmlin{1} =\Bqlin
< \Bpol<\Bqpol<\Bhex.\]

So we have a strict linear ordering between the regular sets introduced so far.
In the next section, we shall show how the comparability of regular sets with respect to $\leq$ can be enforced by a suitable restriction of the universe of time bounds. Of course, the remaining subuniverse should enable us to represent all those regular sets which yield essential complexity classes.
More precisely, we say that a bound set $B$ is {\it represented in}\/ a universe (i.e., a set of bounds) $U'$ iff there is a subset
\[ B' \, \subseteq \, U' \; \mbox{ with } \; B\eqa B'.\]
Then we have
$\XT(B)  = \XT(B')$ for all prefixes $\mbox{\rm X}\in\{\mbox{\rm D}, \mbox{\rm N}, \mbox{\rm A}\}$.
This means that the complexity classes related to $B$ can also be obtained by the subset $B'$ of $U'$.

The following lemma shows that, in order to represent all regular subsets of an arbitrary universe of time bounds, this can be restricted to its time-constructible elements.

\begin{lemma}
Let $U$ be a set of bounds and $B$ be a regular subset of $U$.
Then $B$ is represented in the subuniverse $U'=\{\beta\in U: \beta \mbox{ is time-constructible}\,\}$.
\end{lemma}

\begin{proof} For a regular $B\subseteq U$, let
\[B'= \{\beta'\in U': \: \mbox{ there are $\beta_1,\beta_2\in B$ such that $\beta_1\leqa\beta'\leqa\beta_2$}\,\}.\]
We have to show that $B\equiv B'$.
If $\beta'\in B'$, there is a $\beta_2\in B$ with $\beta'\leqa \beta_2$, hence it holds $B'\leqa B$.
For $\beta\in B$, by condition i) of Definition 1, there is a time-constructible $\beta'\in B\subseteq U$ with $\beta\leqa\beta'$.
From $\beta\leqa\beta'\leqa\beta'$, it follows that $\beta'\in B'$. So we have shown that $B\leqa B'$.
\qed
\end{proof}

The next lemma shows that all sets representing regular sets in the sense of Lemma 5 are themselves regular.

\begin{lemma}
If a set $B$ consists of time-constructible bounds only and is confinal to a regular set, then $B$ is regular, too.
\end{lemma}

\begin{proof}
Let $B\eqa B_1$ for some regular set $B_1$. Notice that both $B_1$ and $B$ have to be non-empty.
Since all $\beta\in B$ are time-constructible, condition i) of Definition 1 holds trivially.
For any $\beta,\beta'\in B$, there are $\beta_1,\beta_1'\in B_1$ such that $\beta\leqa\beta_1$ and $\beta'\leqa\beta_1'$. By Lemma 1 and due to the regularity of $B_1$, it follows that
$\beta+\beta'\circ\beta \leqa \beta_1+\beta_1'\circ\beta_1\leqa \beta_1''$ for some $\beta_1''\in B_1$.
Since $B_1\leqa B$, there is a $\beta''\in B$  with $\beta_1''\leqa \beta''$. This proves condition ii) of regularity for the set $B$.
\qed
\end{proof}

The set of all bound functions, i.e., the maximal {\it universe of bounds}\/,
\[ U\imax = \{ \beta:\, \mbox{$\beta$ is a time bound}\,\},\]
has the cardinality $2^{\aleph_0}$ of the continuum.
Lemmas 5 and 6 show that, in order to represent all regular sets of bounds, we can restrict ourselves to regular subsets of the 
countable sub-universe of the time-constructible bounds,
\[ U'\imax = \{ \beta:\, \mbox{$\beta$ is a time-constructible bound}\,\}.\]

\section{Tameness and regularity}

To guarantee a linear ordering between all (confinality classes of) regular sets, 
it is useful to restrict oneself to universes $U$ of bound functions whose orders of growth are pairwise comparable with each other in a strong manner.
 This leads to the following notion.

\begin{definition}[Tameness]
A set $U$ of bounds is said to be tame iff for any two $\beta_1,\beta_2\in U$ the limit of the sequence of the quotients exists, i.e.,
\[\lim_{n\to \infty} \,\mbox{\normalsize $\frac{\beta_1(n)}{\beta_2(n)}$}\, =\,\limsup_{n\to \infty} \,\mbox{\normalsize $\frac{\beta_1(n)}{\beta_2(n)}$}\,
=\,\liminf_{n\to \infty} \,\mbox{\normalsize $\frac{\beta_1(n)}{\beta_2(n)}$}\,
\in \:\R \,\cup\,\{\,\infty\,\}.\]
\end{definition}

Notice that for all $\beta_1,\beta_2\in U$ in a tame universe $U$ we have
$\beta_1 \leqo\beta_2$ iff $\,\lim_{n\to \infty} \,\mbox{\normalsize $\frac{\beta_1(n)}{\beta_2(n)}$}<\infty$ iff $\,\lim_{n\to \infty} \,\mbox{\normalsize $\frac{\beta_2(n)}{\beta_1(n)}$}>0$, and it holds
$\beta_1 \leo\beta_2$ iff $\,\lim_{n\to \infty} \,\mbox{\normalsize $\frac{\beta_1(n)}{\beta_2(n)}$}=0$ iff $\,\lim_{n\to \infty} \,\mbox{\normalsize $\frac{\beta_2(n)}{\beta_1(n)}$}=\infty$.
Moreover, $\beta_1 \leo\beta_2$ is then equivalent to $\beta_1 \leqo\beta_2$ but not $\beta_2 \leqo\beta_1$, the latter will also be written as $\beta_2 \nleqo\beta_1$.

It is easily seen that the union of all regular sets introduced so far,
\[U_0\, =\, \Blin \,\cup\, \Blogs \,\cup\ \bigcup\nolimits_{m\in\Nplusi} \! \Bqmlin{m} \, \cup \,\Bpol\,\cup \,\Bqpol \,\cup\,\Bhex\,,\]
is a tame set of time-constructible bounds. This remains valid for many supersets obtained from $U_0$ by adding several (sets of) natural bounds which are usually considered in structural complexity theory.

Any two subsets of a tame universe, are comparable with respect to $\leqo$:

\begin{lemma}
For any non-empty subsets $B_1,B_2\subseteq U$ of a tame universe $U$, we have $B_1\leqo B_2$ or $B_2\leqo B_1$.
\end{lemma}
\begin{proof}
We conclude indirectly. If neither $B_1\leqo B_2$ nor $B_2\leqo B_1$,
there is a $\beta_1^0\in B_1$ such that $\beta_1^0 \nleqo \beta_2$, hence $\beta_2 \leo \beta_1^0$, for all $\beta_2\in B_2$ because of the tameness of $U$.
Analogously, there is a $\beta_2^0\in B_2$ such that $\beta_2^0 \nleqo \beta_1$, i.e., $\beta_1 \leo \beta_2^0$, for all $\beta_1\in B_1$.
It follows
$\,\lim_{n\to \infty} \frac{\beta_2^0(n)}{\beta_1^0(n)}= 0 =
\lim_{n\to \infty} \frac{\beta_1^0(n)}{\beta_2^0(n)}$, a contradiction.
\qed
\end{proof}

In particular, any two regular subsets $B_1$ and $B_2$ of a tame universe $U$ are comparable with each other with respect to $\leq$ (meaning $\leqa$ or $\leqo$, see Lemma 4).

Since $\beta\leqa\beta\circ\beta$ for any time bound $\beta$, we have
\[B\,\leqa \,\{\beta\circ\beta:\,\beta\in B\} \,\subseteq \,\{\beta_1\circ\beta_2:\: \beta_1,\beta_2\in B\}\]
for any non-empty set $B$ of bounds. Moreover, for time bounds $\beta$ and $\beta'$, it always holds that $\beta'\circ\beta\leqa \beta+\beta'\circ\beta$.
So, by condition ii) of Definition 1, we obtain
\[\{\beta_1\circ\beta_2:\: \beta_1,\beta_2\in B\}\,\leqa \,B\]
for every regular set $B$.
Conversely, this condition implies the regularity for all sets $B$ consisting of time-constructible bounds and satisfying $\Blin\leqa B$. We here omit the proof of this assertion and proceed with a stronger result.
It says that if a certain tameness is supposed in addition, the regularity of $B$
even follows from a weaker condition.

\begin{proposition}[Criterion of regularity]
Let $\Blin\leqa B$ for a set $B$ of time-constructible bounds such that $B\cup\{\beta_{\iid}\}$ is tame. Then $B$ is regular iff $\:\{\beta\circ\beta:\,\beta\in B\}\,\leqa\, B$.
\end{proposition}

\begin{proof}
By the preceding remark, we have
$\{\beta\circ\beta:\,\beta\in B\}\subseteq \{\beta_1\circ\beta_2:\: \beta_1,\beta_2\in B\}\leqa B$
for any regular $B$.
Now it is shown that, under the further suppositions from Proposition 1, $\{\beta\circ\beta:\,\beta\in B\}\leqa B$ implies condition ii) of regularity of $B$. For any given bounds $\beta,\beta'\in B$, we consider the following three possible cases. \\
Case 1: $\beta'\leqo \beta$ and $\beta\leqo\beta_{\iid}$.
\\
Then we also have $\beta'\leqo\beta_{\iid}$, and it follows $\beta+\beta'\circ\beta\leqo\beta_{\iid}$,
hence $\beta+\beta'\circ\beta\leqo\beta_0$ for a suitable $\beta_0\in B$.
\\
Case 2: $\beta'\leqo \beta$ and $\beta\nleqo\beta_{\iid}$. \\
Since $B\cup\{\beta_{\iid}\}$ is tame and $\beta\nleqo\beta_{\iid}$, we have $\beta_{\iid}\leo\beta$.
This means that $\beta$ is superlinear. Now it follows
\[ \beta+\beta'\circ\beta\,\leqa\, 2\cdot \beta\circ\beta \,\leqa\, (\beta\circ\beta)\circ (\beta\circ\beta).\]
By supposition, there are a $\beta_1\in B$ such that $(\beta\circ\beta)\leqa \beta_1$ as well as a
$\beta_2\in B$ such that $(\beta_1\circ\beta_1)\leqa \beta_2$. So we have  $\beta+\beta'\circ\beta\leqa \beta_2\in B$.
\\
Case 3: $\beta'\nleqo \beta$, thus $\beta\leo\beta'$ due to the tameness of $B\cup\{\beta_{\iid}\}$. \\
Then $\beta'$ is superlinear, and we have $\beta\leqa\beta'$ and $\beta'\circ\beta\leqa\beta'\circ\beta'\leqa\beta_1$,
with a suitable $\beta_1\in B$. Finally we get
\[ \beta+\beta'\circ\beta\,\leqa \,2\cdot \beta'\circ\beta' \,\leqa\, (\beta'\circ\beta')\circ (\beta'\circ\beta')\,\leqa \,\beta_2,\]
for some $\beta_2\in B$.
\qed
\end{proof}
In particular, the tameness of $B\cup\{\beta_{\iid}\}$ enforces that any bound $\beta\in B$
is either superlinear (i.e., $\beta_{\iid}\leo\beta$) or linearly bounded (i.e., $\beta\leqo\beta_{\iid}$).

It should be noticed that the suppositions of Proposition 1, namely the time-constructibility of all elements of $B$,
$\Blin\leqa B$ and the tameness of $B\cup\{\beta_{\iid}\}$, are fulfilled by a large variety of sets of bounds.
So the criterion of regularity given by Proposition 1 is useful in many cases concerning natural universes of time bounds. We still mention that the tameness of $B\cup\{\beta_{\iid}\}$ is equivalent to that of $B\cup\Blin$:
\begin{lemma}
For any set $B$ of bounds, $B\cup\{\beta_{\iid}\}$ is tame iff $B\cup\Blin$ is tame.
\end{lemma}
\begin{proof}
The if-direction is trivial, since $B\cup\{\beta_{\iid}\}\subseteq B\cup\Blin$.
Now let $B\cup\{\beta_{\iid}\}$ be tame. Then both $B$ and $\Blin$ are tame. Thus, for $\beta_1,\beta_2\in B$ or
$\beta_1,\beta_2\in \Blin$ the sequence of quotients has a limit. If $\beta_1\in B$ and $\beta_2\in\Blin$, say $\beta_2(n)= c\cdot n$ with $c\in\Nplus$, then
\[\lim_{n\to \infty} \,\mbox{\normalsize $\frac{\beta_1(n)}{\beta_2(n)}$}\, =\,
  \lim_{n\to \infty} \,\mbox{\normalsize $\frac{\beta_1(n)}{c\cdot n}$}\,=\,
  \mbox{\normalsize $\frac{1}{c}$}\cdot\lim_{n\to \infty} \,\mbox{\normalsize $\frac{\beta_1(n)}{n}$},\]
and the latter limit exists due to the tameness of $B\cup\{\beta_{\iid}\}$.
\qed
\end{proof}

For a time bound $\beta$ and a set $B$ of bounds, respectively, we consider the {\it sets of iterations}\/:
\[ \It(\beta)= \{\beta\ite{m}:\, m\in\Nplus\} \en \mbox{ and } \en
\It(B)= \bigcup\nolimits_{\beta\in B} \! \It(\beta).\]

Since $\beta_{\iid}\leqa \beta$ for any bound $\beta$, Lemma 1 yields $\beta\ite{l}\leqa\beta\ite{k}$
 whenever $l\leq k$. For example, it follows
$\It(\beta)\,\eqa \, \{\beta\ite{2^m}:\, m\in\N\}$ and even $\It(\beta)\,\eqa \, \{\beta\ite{m}:\, m\in M\}$
for every infinite subset $M\subseteq\N$.

If $\{\beta\circ\beta:\,\beta\in B\}\leqa B$, then $\{\beta\ite{2^m}:\, m\in\N, \, \beta\in B\}\leqa B$, hence
we have $\It(B)\leqa B$. So Proposition 1 immediately yields
\begin{corollary}
Let $\Blin\leqa B$ for a set $B$ of time-constructible bounds such that $B\cup\{\beta_{\iid}\}$ is tame. Then $B$ is regular iff $\:\It(B)\,\leqa \,B$.
\qed
\end{corollary}

Since the set operator $\,\It\,$ is idempotent, i.e., it always holds $\It(\,\It(B)\,)=\It(B)$,
for any set $B$ fulfilling the suppositions of Proposition 1 (or Corollary 1) the set $\It(B)$ is regular.

For any superlinear and time-constructible bound $\beta$, the set $B=\It(\beta)$ fulfills these suppositions. The required tameness of $\It(\beta)\cup\{\beta_{\iid}\}$ follows by Lemma 3 which yields: \[\beta_{\iid}\,\leo\,\beta\,\leo\,\beta\circ\beta\,\leo \,\beta\ite{3}\,\leo \,\beta\ite{4}\,
\leo \:\ldots\;. \]
If a bound $\beta$ is {\it ultimately linear}\/, i.e., there is a $c\in\Nplus$ such that $\beta(n)=c\cdot n$ for almost all $n\in\N$,
then $\It(\beta)\cup\{\beta_{\iid}\}$ is tame, too.
So we have
\begin{lemma}
Let $\beta$ be a time-constructible bound which is superlinear or ultimately linear. Then
 $\It(\beta)\cup\{\beta_{\iid}\}$ is tame, and $\,\It(\beta)\,$ is a regular set.
\qed
\end{lemma}

The regular sets of bounds we considered at the beginning of Section 2, $\Bpol,\Blin$ etc., are not sets of iterations of single bounds, but obviously they are confinal to such ones:
\bea
\Blin & \eqa &\It(\beta) \;\mbox{ for } \beta(n)=2\cdot n,\\
\Bpol & \eqa &\It(\beta) \;\mbox{ for } \beta(n)=n^2,\\
\Bqmlin{m} & \eqa &\It(\beta) \;\mbox{ for } \beta(n)=n\cdot (\log\ite{m}(n)),\\
\Blogs& \eqa &\It(\beta) \;\mbox{ for } \beta(n)=n\cdot \log^*(n),\\
\Bqpol & \eqa &\It(\beta_1) \:\eqa \:\It(\beta_2) \;\mbox{ for } \beta_1(n)=n^{\log(n)}
                                                \mbox{ and } \beta_2(n)=n^{(\log(n))^2},\\
\Bhex & \eqa &\It(\beta) \;\mbox{ for } \beta(n)=2^n\,.
\eea
The regularity of $\Blin$ is obvious; the regularities of all the other sets of bounds we just mentioned follow by Lemmas 9 and 6.

Lemma 9 enables us to specify a lot of further regular sets of bounds.
In order to obtain universes of regular sets which are linearly ordered with respect to $\leq$, we shall employ some further tools that will be prepared in the next section.

Here, we still observe that the tameness of a set can be lost in applying the operator of iteration.
For example, let
\[ \beta_1(n) =\left\{ \begin{array}{rl}
                n^2 & \mbox{ if $n$ is even,}\\
                2n^2 & \mbox{ if $n$ is odd,}
                \end{array} \right.
                \qquad \mbox{ and }\en
                \beta_2(n)=n^4 \en\mbox{ for all $n\in\N$}. \]
The set $B=\{\beta_{\iid}, \beta_1,\beta_2\}$ is tame, but $\It(B)$ is not tame. Indeed,
\[ \beta_1\ite{2}(n) =\left\{ \begin{array}{rl}
                n^4 & \mbox{ if $n$ is even,}\\
                4n^4 & \mbox{ if $n$ is odd.}
                \end{array} \right. \]
So we have
\[\limsup\nolimits_{n\to \infty} \,\mbox{\normalsize $\frac{\beta_1^{\langle 2\rangle}(n)}{\beta_2(n)}$}\,=\,4\,\not=\, 1\,=\,
\liminf\nolimits_{n\to \infty} \,\mbox{\normalsize $\frac{\beta_1^{\langle 2\rangle}(n)}{\beta_2(n)}$}.\]
Even if $\It(B)$ is not tame in this example, it is confinal to rather simple tame sets:
\[\It(B)\,\equiv\,\It(\beta_1)\,\equiv\,\It(\beta_2).\]

\section{About iterations and useful properties of universes}

For time bounds $\beta_1$ and $\beta_2$ let $\beta_1\leqi\beta_2$ mean that $\It(\beta_1)\leqa\It(\beta_2)$.
$\leqi$ is a reflexive and transitive relation between time bounds. By $\eqi$ we denote the related equivalence relation:
\[\beta_1\,\eqi\,\beta_2 \en \mbox{ iff } \en \It(\beta_1)\,\eqa\,\It(\beta_2). \]
$\beta_1$ and $\beta_2$ are called {\it it-equivalent}\/ in this case.
By $\beta_1\lei\beta_2$, we mean that $\beta_1\leqi\beta_2$ but not $\beta_1\eqi\beta_2$. 

For example, any two iterations of a single bound $\beta$ are it-equivalent by Lemma 1,
\[ \beta\ite{l} \,\eqi \, \beta\ite{k} \en \mbox{ for all } l,k\in\Nplus\,.\]

Again by Lemma 1, $\beta_1\leqa\beta_2$ implies $\beta_1\ite{m}\leqa\beta_2\ite{m}$ for all $m\in\Nplus$, hence
$\It(\beta_1)\leqa\It(\beta_2)$, i.e., $\beta_1\leqi\beta_2$. Relation $\leqi$ is also a weakening of $\leqo$,
at least for superlinear bounds:
\begin{lemma}
If $\beta_1\leqo\beta_2$ and $\beta_2$ is a superlinear bound, then $\beta_1\leqi\beta_2$.
Thus, $\beta_1\lei\beta_2$ implies $\beta_2\nleqo\beta_1$, for superlinear $\beta_2$.
\end{lemma}
\begin{proof}
$\beta_1\leqo\beta_2$ means that $\beta_1\leqa c\cdot \beta_2$ with a suitable constant $c\in\Nplus$.
If $\beta_2$ is superlinear, Lemma 3 yields $\beta_2\leo \beta_2\circ\beta_2$, hence
$c \cdot\beta_2\leo \beta_2\circ\beta_2$ and $\beta_1\leqa\beta_2\circ\beta_2$. By Lemma 1, it follows
 $\beta_1\ite{m}\leqa\beta_2\ite{2m}$ for all $m\in\Nplus$, thus $\It(\beta_1)\leqa\It(\beta_2)$.
\qed
\end{proof}

One should notice the special role of the identical bound $\beta_{\iid}$ with respect to it-equivalence.
Since $\It(\beta_{\iid})=\{\beta_{\iid}\}$, even the set $\Blin$ of linear bounds yields two equivalence classes with respect to $\eqi$, namely $\{\beta_{\iid}\}$ and $\Blin\setminus \{\beta_{\iid}\}$. For all 
$\beta\in\Blin\setminus \{\beta_{\iid}\}$ we have $\beta\leqo \beta_{\iid}$ but not $\beta\leqi \beta_{\iid}$.

\begin{definition}[it-embeddability and it-isomorphism]
For universes (i.e., sets) of bounds, $U_1$ and $U_2$, we say that $U_1$ is it-embeddable in $U_2$,
in symbols: $U_1 \emi U_2$, iff for any $\beta_1\in U_1$ there is a $\beta_2\in U_2$ such that $\beta_1\eqi\beta_2$.\\
$U_1$ and $U_2$ are said to be it-isomorphic iff they are mutually it-embeddable in each other:
\[ U_1 \isi U_2 \en \mbox{ iff }\en U_1 \emi U_2 \mbox{ and }\, U_2 \emi U_1. \]
\end{definition}

\begin{lemma}
If $U_1 \emi U_2$ and $U_2$ consists of time-constructible bounds only,
then every regular subset $B\subseteq \It(U_1)$ can be represented in $\It(U_2)$.
\end{lemma}

\begin{proof}
Let $U_1 \emi U_2$ and $B$ be a regular subset of $\It(U_1)$.
By the Axiom of Choice, there is a mapping $\fct{f}{B}{U_2}$ such that $f(\beta)\eqi\beta$ for all $\beta\in B$.
Now we put
\[B'\,=\, \It(\,\{f(\beta):\: \beta\in B\}\,).\]
For any $\beta\in B$, it holds $\It(f(\beta))\eqa\It(\beta)$, hence there is a $\beta'\in\It(f(\beta))$ with $\beta\leqa\beta'$. So we have $B\leqa B'$.

For any $\beta'\in B'$, there is a $\beta\in B$ such that $\beta'\in\It(f(\beta))$.
Since $\It(f(\beta))\eqa\It(\beta)$, there is a $\beta_0\in\It(\beta)$ with $\beta'\leqa\beta_0$.
Remember that $\It(B)\leqa B$ for every regular set $B$. Thus, there is a $\beta_1\in B$ with $\beta_0\leqa\beta_1$, and we have shown that $B'\leqa B$.

So we have $B'\eqa B$. By supposition of the lemma, $U_2$ consists of time-constructible bounds. Then this holds for
$\It(U_2)$, too, and so for its subset  $B'$. Lemma 6 shows that $B'$ is regular.
\qed
\end{proof}

By Lemma 11, if $U_1$ and $U_2$ are it-isomorphic universes both consisting of time-constructible bounds only, then
their iteration sets, $\It(U_1)$ and $\It(U_2)$,
yield systems of regular sets which are equivalent with respect to confinality. In particular, it immediately follows

\begin{proposition}
If $U_1 \isi U_2$ for universes $U_1$ and $U_2$ of time-constructible bounds, then
\[ \{\XT(B):\: B\subseteq \It(U_1) \mbox{ and $B$ regular}\,\}
 \, =\, \{\XT(B):\: B\subseteq \It(U_2) \mbox{ and $B$ regular}\,\}\]
for each prefix $\mbox{\rm X}\in\{\mbox{\rm D}, \mbox{\rm N}, \mbox{\rm A}\}$.
\qed
\end{proposition}

In many cases,
Lemma 11 and Proposition 2 even hold with the universes $U_i$ instead of $\It(U_i)$, for $i=1,2$.
For concise formulations, we introduce the notion of it-completeness of a universe $U$.
This is a weakening of the closedness under the set operator $\It$, which would mean that $\It(U)\subseteq U$.
Also a weakening of tameness will be employed in the sequel.

\begin{definition}[it-completeness and it-tameness]
A universe $U$ of bounds is called it-complete iff for any $\beta\in U$ there is a subset $B\subseteq U$ such that
$B\eqa \It(\beta)$.\\
$U$ is said to be it-tame iff any two $\beta_1,\beta_2\in U$ are comparable with respect to $\leqi$, i.e., it holds
$\beta_1\leqi\beta_2$ or $\beta_2\leqi\beta_1$.
\end{definition}

\begin{corollary}
Let $U_1$ and $U_2$ be it-complete sets of time-constructible bounds.\\
If $U_1 \emi U_2$, then every regular subset $B\subseteq U_1$ can be represented in $U_2$.
If $U_1 \isi U_2$, then
\[ \{\XT(B):\: B\subseteq U_1 \mbox{ and $B$ regular}\,\}
 \, =\, \{\XT(B):\: B\subseteq U_2 \mbox{ and $B$ regular}\,\}\]
for each prefix  $\mbox{\rm X}\in\{\mbox{\rm D}, \mbox{\rm N}, \mbox{\rm A}\}$.
\end{corollary}

\begin{proof}
The first assertion is shown by a slight modification of the proof of Lemma 11,
and it implies the second one.
\qed
\end{proof}

For example, the union of all regular sets introduced in Section 2
\[U_{00}\, =\, \Blin \,\cup\, \Blogs \,\cup\ \bigcup\nolimits_{m\in\Nplusi} \! \Bqmlin{m} \, \cup \,\Bpol\,\cup \,\Bqpol \,\cup\,\Bhex\,,\]
is it-complete. So, even if it does not contain all iterations of its elements, i.e., $U_{00}\not= \It(U_{00})$,
it is sufficient to represent all regular subsets of $\It(U_{00})$.
Indeed, we easily see that $U_{00}\,\isi\,\It(U_{00})$.

 $U_{00}$ is both tame and  it-tame. It can also be represented by a subuniverse of bounds which are strictly linearly ordered with respect to $\leqi$, as the next result shows.

\begin{proposition}
Let $U$ be an it-tame universe of bounds. Then there is an it-isomorphic subuniverse $U'\isi U$ such that
for any two $\beta_1,\beta_2\in U'$, it holds
\[ \beta_1=\beta_2 \en \mbox{ or }\en \beta_1 \lei \beta_2 \en \mbox{ or }\en \beta_2 \lei  \beta_1. \]
\end{proposition}

\begin{proof}
Let $U'$ be any complete (but minimal) choice of bounds representing the different $\eqi$-classes of elements of $U$.
Then it is linearly ordered by $\lei$.
Since $U'\subseteq U$, we have $U'\emi U$. Conversely, for any $\beta\in U$ there is exactly one $\beta'\in U'$ with
$\beta \eqi \beta'$. Thus $U\emi U'$.
\qed
\end{proof}

\section{A dense universe of bounds between linear and polynomial time}

Due to Proposition 1 and Corollary 1, in order to simplify the check for regularity of a set of bounds, it is desirable to restrict the universe of bounds to such functions $\beta$ for which $\It(\beta)$ can arithmetically be characterized in some easy way. In the present section, we shall deal with a special type of such bounds. They are obtained by means of nondecreasing functions $\fct{\alpha}{\N}{\N}$ satisfying 
$1\leqa\alpha(n)\leqa n$.

For arbitrary functions $\fct{\alpha_1,\alpha_2}{\N}{\N}$, let $\alpha_1\,\leqpow\, \alpha_2$ mean that
for each $k\in\N$ there is an $l\in\N$ such that $\,\alpha_1^k\,\leqa\,\alpha_2^l$, i.e., $(\alpha_1(n))^k\leq(\alpha_2(n))^l$ for
almost all $n\in\N$.

The corresponding equivalence relation $\,\eqpow\,$ as well as  the strict ordering $\,\lepow\,$ are straightforwardly defined:
$\alpha_1\,\eqpow\, \alpha_2$ means that both $\alpha_1\,\leqpow\, \alpha_2$ and $\alpha_2\,\leqpow\, \alpha_1$; 
$\alpha_1\,\lepow\, \alpha_2$ indicates that $\alpha_1\,\leqpow\, \alpha_2$ but not  $\alpha_2\,\leqpow\, \alpha_1$. 

Another sharpening of $\leqpow$ is the relation $\,\llpow\,$;
 it will later be crucial in showing the tameness of certain universes of bounds. 
By $\alpha_1 \llpow\alpha_2$ we mean that $\alpha_1^k\leqa\alpha_2$ for all $k\in\N$.

\begin{lemma}
If $2\,\leqa\, \alpha_1\,\llpow\,\alpha_2$, then we have $\alpha_1^k \,\leo\,\alpha_2$ for all $k\in\N$,
and it holds neither  $\alpha_2\,\leqo\,\alpha_1\,$ nor $\,\alpha_2\,\leqa\,\alpha_1$.
\end{lemma}

\begin{proof}
The supposition implies that
$2^l\cdot \alpha_1^k \leqa \alpha_1^{l+k} \leqa \alpha_2$ for all $l,k\in\N$.
It follows $\frac{2^l\cdot \alpha_1^k(n)}{\alpha_2(n)}\,\leqa\,1$, 
hence $\frac{ \alpha_1^k(n)}{\alpha_2(n)}\,\leqa\,\frac{1}{2^l}$.
Thus, we have
$\lim_{n\to\infty}\frac{ \alpha_1^k(n)}{\alpha_2(n)} =0$, and this means $\alpha_1^k \,\leo\,\alpha_2$.

From   $\alpha_2\,\leqo\,\alpha_1$ or $\alpha_2\,\leqa\,\alpha_1$ it would follow that 
$\alpha_2\leqa c\cdot \alpha_1$ for some $c\in\Nplus$. This yields $\frac{\alpha_2(n)}{\alpha_1(n)}\,\leqa\, c$, a contradiction to 
$\alpha_1 \,\leo\,\alpha_2$.
\qed
\end{proof}

For example, we have $\alpha_1\eqpow\alpha_2$ for any two polynomials $\alpha_i(n)=n^{e_i}$ with $e_i\in\Nplus$ (for $i=1,2$).
It even holds $(\alpha(n))^e \eqpow \alpha(n)$ for all exponents $e\in\Nplus$ and arbitrary functions $\alpha$, but
 $\log\ite{m_1}(n) \lepow \log\ite{m_2}(n)$ and  $\log\ite{m_1}(n) \llpow \log\ite{m_2}(n)$ 
for all $m_1>m_2\geq 1$.

The {\it set of powers}\/ of a function $\alpha$ is defined by
\[ \Pow(\alpha)= \{ \alpha^m:\, m\in\Nplus\}.\]

The following property of f-consistency ensures that the function $\alpha$ taken as a factor in defining
$\beta(n)= n\cdot \alpha(n)$ yields a well controllable iteration set $\It(\beta)$. 

\begin{definition}[f-consistency]
A function $\fct{\alpha}{\N}{\N}$ is said to be f-consistent iff 
\[\{\alpha(n^k):\, k\in\N\} \,\leqa \,\Pow(\alpha).\]
 This means that to any $k\in\N$ there is an $l\in\N$ such that
$\alpha(n^k)\, \leqa\, (\alpha(n))^l \,$,\en  i.e.,
$\alpha(n^k) \leq (\alpha(n))^l \,$ for almost all $ n\in\N$.
\end{definition}
If $\alpha$ is supposed to be nondecreasing, then it is already f-consistent if there is an  $l\in\N$ such that
$\alpha(n^2)\,\leqa\, (\alpha(n))^l$.

The following lemma provides a useful arithmetical characterization of $\It(\beta)$ for a variety of time bounds $\beta$ between the linear and the polynomial ones.

\begin{lemma}
Let $\beta(n)= n\cdot \alpha(n)$ with a nondecreasing, f-consistent function $\alpha$ such that $1\,\leqa\, \alpha(n)$
and $\alpha(n)\,\leqa\, n$. Then 
\[\It(\beta)\, \eqa\, \{\, n\cdot \alpha'(n): \,\alpha'\in\Pow(\alpha)\,\}.\]
If $\beta_i(n)=n\cdot\alpha_i(n)$ with  nondecreasing, f-consistent functions $\alpha_i$ satisfying 
$1\,\leqa\, \alpha_i(n) \,\leqa\, n$ for $i=1,2$, then 
\[\beta_1\,\leqi\,\beta_2 \en \mbox{ iff } \en \alpha_1\,\leqpow\,\alpha_2 \quad \mbox{ and } \quad 
\beta_1\,\lei\,\beta_2 \en \mbox{ iff } \en \alpha_1\,\lepow\,\alpha_2,\]
and from $\,2\,\leqa\, \alpha_1\,\llpow\,\alpha_2\,$ it follows that $\,\beta_1\,\lei\,\beta_2\,$.
\end{lemma}

\begin{proof}
To prove the inequality $\It(\beta)\, \leqa\, \{\, n\cdot \alpha'(n): \,\alpha'\in\Pow(\alpha)\,\}$,
we inductively show that for all $k\in\Nplus$ there is an $l\in\N$ such that 
\[\beta\ite{k}(n)\,\leqa\, n\cdot(\alpha(n))^l.\]
 This holds for $k=1$ and $l=1$ by definition of $\beta$. If it holds for some $k$ and $l$,
then 
\bea
 \beta\ite{k+1}(n) &=& \beta\circ\beta\ite{k}(n) \\
   &=& \beta\ite{k}(n)\cdot \alpha(\beta\ite{k}(n))\\
&\leqa& n\cdot(\alpha(n))^l\cdot \alpha(n\cdot(\alpha(n))^l)\\
& \leqa & n\cdot(\alpha(n))^l\cdot \alpha(n^{l+1})\\
&\leqa & n\cdot(\alpha(n))^l\cdot( \alpha(n))^{l'}\\
& \leq & n\cdot(\alpha(n))^{l+l'},
\eea
with a suitable $l'\in\N$. Thus, the step of induction has been done.

For the converse inequality, it is inductively shown that 
\[n\cdot(\alpha(n))^k \leqa \beta\ite{k}(n).\]
This holds for $k=1$, and the step of induction follows by
\bea
 n\cdot (\alpha(n))^{k+1} &=& n\cdot(\alpha(n))^k\cdot \alpha(n) \\
& \leqa & \beta\ite{k}(n)\cdot \alpha(n) \\
& \leqa & \beta\ite{k}(n)\cdot \alpha(\beta\ite{k}(n)) \\
& = & \beta\circ\beta\ite{k}(n)\en = \en \beta\ite{k+1}(n).
\eea
So the stated representation of $\It(\beta)$ has been shown. The second assertion of the lemma follows immediately,
whereas Lemma 12 yields the last one. 
\qed
\end{proof}

In the sequel, we employ an integer square root  function and a slight modification of the integer logarithm: let
\[\msqrt(n) =\lfloor \sqrt{n} \rfloor \en \mbox{ and } \en \mlog(n)= \left\{
   \begin{array}{cl}\lfloor \log_2(n)\rfloor & \mbox{ if } n\geq 1 \\
                                0 & \mbox{ if } n=0.
 \end{array} \right. \]
For example, it holds $\msqrt(n\cdot m) = \lfloor \sqrt{n}\cdot \sqrt{m}\rfloor \leq \sqrt{n}\cdot\sqrt{m}\leq 
(\msqrt(n)+1)\cdot(\msqrt(m)+1)$. Also, we have $\Pow(\msqrt\circ \alpha)\eqa\Pow(\alpha)$
for all nondecreasing unbounded functions $\fct{\alpha}{\N}{\N}$.

\begin{lemma}
Let $\alpha_1$ and $\alpha_2$ be nondecreasing, f-consistent functions with
 $1\,\leqa \,\alpha_i(n)\leqa n$ for $i=1,2$.
Then the function $\alpha$ defined by
\[\alpha(n) = 2^{\msqrti(\,(\mlogi\,\circ\,\alpha_1(n))\cdot(\mlogi\,\circ \,\alpha_2(n))\,)  }\]
fulfills all these properties, too.
\end{lemma}

\begin{proof}
It is obvious that $1\leq \alpha(n)\leq \max(\alpha_1(n),\alpha_2(n))\leqa n$.
Since the functions $\alpha_1,\alpha_2,\mlog$ and $\msqrt$ as well as $\beta_{\mbox{\footnotesize exp}}$ are nondecreasing and the product and composition of nondecreasing functions remains nondecreasing, it follows that $\alpha$ is nondecreasing.
In order to prove the f-consistency, it is enough to deal with the exponent $k=2$. We have
\bea
\alpha(n^2) &=&  2^{\msqrti(\,(\mlogi\circ\alpha_1(n^2))\cdot(\mlogi\circ \alpha_2(n^2))\,)  }\\
                    &\leq&  2^{\msqrti(\,\mlogi(l_1\cdot\alpha_1(n))\cdot\mlogi(l_2\cdot \alpha_2(n))\,)  }\\
                    &\leq&  2^{\msqrti(\, l_1\cdot l_2 \cdot (\mlogi\circ\alpha_1(n))\cdot(\mlogi\circ \alpha_2(n))\,)  }\\
                    &\leq&  2^{(\msqrti(l_1\cdot l_2)+1)\cdot (\msqrti(\,(\mlogi\circ\alpha_1(n))\cdot(\mlogi\circ \alpha_2(n)))+1\,)  }\\
                    &\leq&  2^{(\msqrti(\,(\mlogi\circ\alpha_1(n))\cdot(\mlogi\circ \alpha_2(n)))+1\,) \cdot ( \msqrti(l_1\cdot l_2)+1)  }\\
                    &\leqa& 2^{  \msqrti(\,(\mlogi\circ\alpha_1(n))\cdot(\mlogi\circ \alpha_2(n))\,) \cdot 2\cdot \msqrti(l_1\cdot l_2+1)  }\\
                    &\leqa& 2^{  \msqrti(\,(\mlogi\circ\alpha_1(n))\cdot(\mlogi\circ \alpha_2(n))\,)\, \cdot\, l'  }\\
                    &=& (\alpha(n))^{l'} \en \mbox{ with suitable $l_1,l_2, l'\in\N$}. \hspace*{7.25cm} \qed
\eea
\end{proof}

\begin{lemma}
Let $\alpha_1$ and $\alpha_2$ be nondecreasing functions satisfying
 $1\,\leqa \,\alpha_i(n)\leqa n$ for $i=1,2$, and let $\alpha$ be defined as in Lemma 14.
If $2\leqa \alpha_1\leqpow\alpha_2$, then  it holds $\alpha_1\leqpow\alpha\leqpow\alpha_2$. If, moreover,   $\alpha_1\lepow\alpha_2$, then we even have $\alpha_1\lepow\alpha\lepow\alpha_2$.
From $2\leqa \alpha_1\llpow \alpha_2$ it follows that $\alpha_1\,\llpow\,\alpha\,\llpow\,\alpha_2$.
\end{lemma}

\begin{proof}
In showing the first assertion, without loss of generality,
we can suppose that $\alpha_1\leqa \alpha_2$. Hence, for any $l\in\N$,
\bea
(\alpha_1(n)\,)^l &\leq & 2^{l\cdot(\mlogi\circ\alpha_1(n)+1)}\\
                            &\leqa& 2^{l\cdot\sqrt{\mlogi\circ\alpha_1(n)+1}\cdot \sqrt{\mlogi\circ\alpha_2(n) +1}}\\
                             &\leqa& 2^{2 l\cdot\sqrt{(\mlogi\circ\alpha_1(n))\cdot(\mlogi\circ\alpha_2(n))}}\\ 
                            &\leqa& 2^{2 l\cdot(\msqrti((\mlogi\circ\alpha_1(n))\cdot(\mlogi\circ\alpha_2(n)))+1)}\\
                             &\leqa& 2^{4 l\cdot(\msqrti((\mlogi\circ\alpha_1(n))\cdot(\mlogi\circ\alpha_2(n)))}\\
                           & = &  (\alpha(n))^{4l}
\eea
and
\bea 
 (\alpha(n))^l &=& 2^{ l\cdot(\msqrti((\mlogi\circ\alpha_1(n))\cdot(\mlogi\circ\alpha_2(n)))}\\
                      & \leqa &  2^{ l\cdot(\msqrti((\mlogi\circ\alpha_2(n))^2)}\\
                        & \leqa &  2^{ l\cdot\sqrt{(\mlogi\circ\alpha_2(n))^2}}\\
                     & \leqa &  2^{ l\cdot(\mlogi\circ\alpha_2(n))}\\
                      &\leqa & ( \alpha_2(n))^l.
\eea
Thus, the first assertion of the  lemma has been shown.

If $\alpha_1(n)\lepow\alpha_2(n)$, i.e., $\alpha_2(n)\leqpow \alpha_1(n)$ does not hold,
then $\alpha_2(n) \nleqo (\alpha_1(n))^l$, and for all $c,l\in\N$ we have $\alpha_2(n)> c\cdot (\alpha_1(n))^l$, for
infinitely many numbers $n$.

From $\alpha_1(n)\eqpow \alpha(n)$ it would follow $\alpha\, \leqa\, \alpha_1^k$ for some $k\in\N$, i.e.,
\bea
 2^{\msqrti(\,(\mlogi\,\circ\,\alpha_1(n))\cdot(\mlogi\,\circ \,\alpha_2(n))\,)  } &\leqa & 2^{k\cdot\sqrt{(\log_2\circ\alpha_1(n))^2}},
 \en \mbox{ hence}\\
\msqrt(\,(\mlog\,\circ\,\alpha_1(n))\cdot(\mlog\,\circ \,\alpha_2(n))\,) &\leqa & k\cdot\sqrt{(\log_2\circ\alpha_1(n))^2}.
\eea
Since $\frac{1}{2}\cdot \sqrt{(\frac{1}{2}\log_2\circ\alpha_1(n))\cdot (\frac{1}{2}\log_2\circ\alpha_2(n))} \leqa
\msqrt(\,(\mlog\,\circ\,\alpha_1(n))\cdot(\mlog\,\circ \,\alpha_2(n))\,)$, we would have  
\bea
 \sqrt{\log_2\circ\alpha_2(n)} & \leqa &4k\cdot \sqrt{\log_2\circ\alpha_1(n)}, \en \mbox{thus}\\
\log_2\circ\alpha_2(n)&\leqa& 16k^2\cdot \log_2\circ\alpha_1(n), \en \mbox{and}\\
\alpha_2(n) &\leqa & (\alpha_1(n))^{16 k^2},\en \mbox{a contradiction.}
\eea

From $ \alpha(n)\eqpow \alpha_2(n)$ it would follow that $\alpha_2 \,\leqa\, \alpha^k$ for some $k\in\N$, i.e.,
\bea
2^{\log_2\circ\alpha_2(n)} & \leqa &  2^{k\cdot\msqrti(\,(\mlogi\,\circ\,\alpha_1(n))\cdot(\mlogi\,\circ \,\alpha_2(n))\,)  }\\
                                      &\leqa & 2^{k\cdot\sqrt{(\log_2\circ\alpha_1(n))(\log_2\circ\alpha_2(n))}},  \en \mbox{ hence}\\
\log_2\circ \alpha_2(n) & \leqa &k\cdot \sqrt{(\log_2\circ\alpha_1(n))(\log_2\circ\alpha_2(n))},\\
\sqrt{\log_2\circ\alpha_2(n)} & \leqa&k\cdot \sqrt{\log_2\circ\alpha_1(n)}, \\
\log_2\circ\alpha_2(n) & \leqa&k^2\cdot(\log_2\circ\alpha_1(n)), \en \mbox{ thus}\\
\alpha_2(n) & \leqa &(\alpha_1(n))^{k^2}, \en \mbox{again a contradiction.}
\eea

To show the last assertion of the lemma, let $2\leqa \alpha_1\llpow \alpha_2$, i.e., $\alpha_1^k\leqa \alpha_2$ for all $k\in\N$.
Thus,  $k^2\cdot\log_2\circ\alpha_1 \leqa \log_2\circ\alpha_2$ and we have
\bea 
 (\alpha(n))^k &=& 2^{k\cdot(\msqrti((\mlogi\circ\alpha_1(n))\cdot(\mlogi\circ\alpha_2(n)))}\\
                        & \leqa &  2^{ k\cdot\sqrt{(\log_2\circ\alpha_1(n)) \cdot  (\log_2\circ\alpha_2(n)) }}\\
                   & = &  2^{\sqrt{k^2\cdot (\log_2\circ\alpha_1(n)) \cdot  (\log_2\circ\alpha_2(n)) }}\\
              & \leqa &  2^{ \sqrt{(\log_2\circ\alpha_2(n)) \cdot  (\log_2\circ\alpha_2(n)) }}\\
                     &= &  2^{ \log_2\circ\alpha_2(n)   }\\
                      &= &  \alpha_2(n).
\eea
On the other hand, it holds
\bea 
 (\alpha_1(n))^k &=& 2^{ k\cdot\log_2\circ\alpha_1(n)}\\
               & = &  2^{\sqrt{ k^2\cdot (\log_2\circ\alpha_1(n))^2}}\\
                      & \leqa &  2^{ 2\cdot (\msqrti( k^2 \cdot   (2\cdot\mlogi\circ\alpha_1(n))^2)}\\
                        & \leqa &  2^{\msqrti( 16 k^2 \cdot(\mlogi\circ\alpha_1(n))^2)}\\
            & \leqa &  2^{\msqrti( (\mlogi\circ\alpha_1(n)) \cdot (\mlogi\circ\alpha_2(n))  )}\\
              & = & \alpha(n).      
\eea
This completes the proof of Lemma 15.
\qed
 \end{proof}

\begin{definition}[Bounds of Type 1]
A time bound $\beta$ is said to be of Type  1 iff
\[\beta(n)= n\cdot\alpha(n)\]
with a nondecreasing function $\alpha$ such that $2\leqa\alpha(n)\leqa n$ and $\alpha$ is f-consistent and constructible in linear time.
\end{definition}
Notice that all bounds of Type  1 are time-constructible.
The smallest bound of Type  1 is  almost everywhere equal to the linear bound $\beta(n)=2n$, the largest one is almost everywhere equal to the polynomial $\beta(n)=n^2$.

\begin{theorem}
There is a tame set $U_1$ consisting of time bounds of Type  1 such that 
\[\Blin \,\cup\, \Blogs \,\cup\, \bigcup\nolimits_{m\in\Nplusi} \! \Bqmlin{m} \, \cup \,\Bpol\; \emi \; U_1\]
and $U_1$ is densely linearly ordered by $\,\leqi$. Moreover, if $\beta_i(n)= n\cdot\alpha_i(n)$, for $i=1,2$ and 
$\beta_1,\beta_2\in U_1$, then it holds:
$\;\beta_1\lei\beta_2\:$ iff $\:\alpha_1\llpow\alpha_2$.
\end{theorem}

\begin{proof}
Let 
$\: U_1=  \bigcup\nolimits_{i\in\Ni} W_i\,,\:$
where the sets $W_i$ are successively generated as follows.\\
$W_0$ consists of the functions (given by describing terms):
\[2\,n\,,\, n\cdot \log^*(n)\,,\, n\cdot \log\ite{m}(n), \mbox{ for all $m\in\Nplus$}\,,\mbox{ and }\,n^2.\]
All these bounds are of Type  1, $W_0\subseteq U_1$ implies that 
$\Blin \cup \Blogs\cup \bigcup\nolimits_{m\in\Nplusi} \! \Bqmlin{m} \cup\Bpol\, \emi \,U_1$,  $W_0$ is
linearly ordered with respect to $\,\leqi$, and all $\beta_1,\beta_2\in W_0$ fulfill the logical equivalence stated at the end of  the theorem.

Given a set $W_j$ of time bounds of Type  1 which is linearly ordered by $\,\leqi$, let
\bea
W_{j+1}& =& W_j \,\cup\,\{\beta: \mbox{ there are functions $\beta_1, \beta_2\in W_j$ such that }  \beta_1\,\lei\,\beta_2,\\
&& \hspace*{1.97cm} \mbox{but there is no $\beta_3\in W_j$ with $\beta_1\lei\beta_3\lei\beta_2$},\\
&& \hspace*{1.97cm} \mbox{and }\:
\beta(n) = n\cdot 2^{\msqrti(\,(\mlogi\,\circ\,\alpha_1(n))\cdot(\mlogi\,\circ \,\alpha_2(n))\,)  }, \\
 && \hspace*{5cm} \mbox{where $\beta_i(n)=n\cdot \alpha_i(n)\,$ for $i=1,2$} \,\}.
\eea
By Lemmas 13-15, $W_{j+1}$  is linearly ordered by $\,\leqi$ and  consists of bounds of Type  1 only, where the time-constructibility of the elements can be shown by standard arguments. 

Moreover, from the above construction it 
follows that the universe $U_1$ owns all the properties  stated in the theorem. In particular, the density of the ordering follows by the stepwise construction. 
The equivalence $\beta_1\lei\beta_2$ iff $\alpha_1\llpow\alpha_2$, for any two elements  $\beta_i(n)= n\cdot\alpha_i(n)$,
holds within $W_0$ and follows for all elements of $W_{j+1}$ if it holds in $W_j$. Thus, it holds for $U_1$, too. Therefore, the universe $U_1$ is tame.
\qed
\end{proof}

It should be noticed that even $\It(U_1)$ is tame.
The smallest bound in $U_1$ is  $\beta(n)=2n$, whereas the largest one is $n^2$. 
The sets $U_1\cup\{\beta_{\iid}\}$ and $\It(U_1)\cup\{\beta_{\iid}\}$ are also tame
and linearly ordered by $\,\leqi$, but there is a jump between the functions $\beta_{\iid}$ and $2n$.

If one starts with $\alpha_1(n)=2$ and a bound $\beta_0(n)= n\cdot \alpha_0(n)$ of Type  1 with an unbounded function $\alpha_0$, then the repeated application 
of Lemmas 14 and 15 (with successively changed functions $\alpha_2$) yields a downwards infinite chain $\, \ldots\, \lei \,\beta_2\,\lei\,\beta_1\,\lei\,\beta_0$ of bounds of Type  1, with
$\, \ldots\, \llpow \,\alpha_2\,\llpow\,\alpha_1\,\llpow\,\alpha_0$ for the corresponding factor functions.
Another way to construct such chains is opened by the following result.

\begin{lemma}
Let $\beta_1(n)= n\cdot \alpha_1(n)$ and $\beta_2(n)= n\cdot \alpha_2(n)$ be bounds of Type  1. Then the function
\[\beta(n)=\, n\cdot (\alpha_1\circ\alpha_2(n)\,)\]
is also of Type 1, and we have $\beta\leqi \beta_1$ and $\beta\leqi \beta_2$. If, moreover, $\alpha_2$ is unbounded and 
$\alpha_1\,\llpow\,n$, then  $\alpha_1\circ\alpha_2\,\llpow \,\alpha_2$ and $\: \beta\,\lei\,\beta_2$.
\end{lemma}

\begin{proof}
To show the first assertion, notice that the composition $\alpha_1\circ\alpha_2$ remains nondecreasing, f-consistent and constructible in linear time if both $\alpha_1$ and $\alpha_2$ own these properties and, moreover, $\alpha_1(n),\alpha_2(n)\leqa \, n$.
Also we have
 $\alpha_1\circ\alpha_2(n)\leqa \alpha_1(n),\,\alpha_2(n)$, hence
\[  n\cdot (\alpha_1\circ\alpha_2)(n)\leqi n\cdot \alpha_1(n) \mbox{ and }
 n\cdot (\alpha_1\circ\alpha_2)(n)\leqi n\cdot \alpha_2(n).\]

If not $\alpha_1\circ\alpha_2\llpow \alpha_2$, there would be a $k\in\Nplus$ such that it does not hold
 $(\alpha_1\circ\alpha_2)^k \leqa\alpha_2$, hence we would have $(\alpha_1\circ\alpha_2)^k(n) >\alpha_2(n)$ for infinitely many 
numbers $n$. If $\alpha_2$ is unbounded, this would imply $( \alpha_1(m))^k> m$ for infinitely many $m\in\N$,
hence $\liminf_{n\to\infty} \frac{n}{\alpha_1^k(n)} <1$,
in contradiction to the supposition $\alpha_1\,\llpow\,n$.
So we have $\alpha_1\circ\alpha_2\llpow \alpha_2$, and from this $\beta\,\lei\,\beta_2$ follows by Lemma 13.
\qed
\end{proof}

Applying Lemma 16 to $\alpha_1(n)=\log(n)$ and $\alpha_2(n)=n$, i.e., $\beta_2(n)=n^2$, we obtain $\beta(n)=n\cdot\log(n)$.
This just yields the quasilinear bounds: $\,\It(n\cdot\log(n))\equiv \Bqlin$, whereas $\It(n^2)\equiv \Bpol$.
Applying the construction to $\alpha_1(n) =\log(n)=\alpha_2(n)$, we obtain $\beta(n)= n\cdot \log\ite{2}(n)$ and 
$\It(\beta)\equiv\Bqmlin{2}$.
So one can successively generate all the sets $\Bqmlin{m}$ for $m\in\Nplus$.
Analogously, starting with $\alpha_1(n)=\log(n)$ and $\alpha_2(n)=\log^*(n)$, one generates the downwards infinite chain
$\,\ldots\,\llpow\, \log\ite{2}\circ\log^*(n)\,\llpow\,\log\circ\log^*(n)\,\llpow\,\log^*(n)$.
A start with  $\alpha_1(n)=\log^*(n)=\alpha_2(n)$ yields the downwards infinite chain
$\;\ldots\,\llpow\, \log^*\ite{3}(n)\,\llpow\,\log^*\ite{2}(n)\,\llpow\,\log^*(n)$.

\section{A dense universe of bounds above polynomial time}

The supposition $\alpha(n)\leqa n$ was essential for the proof of Lemma 13. Thus, at most polynomials can be reached by bounds of Type  1, and in order to obtain (a dense universe of) bounds beyond polynomial time, other types of construction are needed.
The following property of e-consistency ensures that a related function $\alpha$, if it is taken as the exponent in defining
$\beta(n)=n^{\alpha(n)}$, yields a controllable iteration set $\It(\beta)$.  

\begin{definition}[e-consistency]
A function $\fct{\alpha}{\N}{\N}$ is said to be e-consistent iff 
\[\{\alpha(n^{(\log(n))^k}):\, k\in\N\}\: \leqa\: \Pow(\alpha).\]
 This means that to every $k\in\N$ there exists an $l\in\N$ such that
$\:\alpha(n^{(\log(n))^k})\, \leqa\, (\alpha(n))^l\,$.
\end{definition}

Examples of e-consistent functions are the powers of the logarithm,
$\alpha(n)= (\log(n))^m$, for arbitrary $m\in\N$.
Indeed, then we have 
\[\alpha(n^{(\log(n))^k})=(\log(n^{(\log(n))^k}))^m\leqa (\,(\log(n))^k +1)^m
 \leqa (\log(n))^{(k+1)\cdot m}=(\alpha(n))^{k+1}.\]

\begin{lemma}
If $\beta(n)=n^{\alpha(n)}$ and $1\,\leqa \,\alpha(n)\,\leqa \log(n)$
with some nondecreasing and e-consistent function $\alpha$. Then
\[ \It(\beta)\,\eqa\,\{ n^{(\alpha(n))^l}: \; l\in\N\}.\]
If $\beta_i(n)=n^{\alpha_i(n)}$ with  nondecreasing, e-consistent functions $\alpha_i$ satisfying 
$1\,\leqa\, \alpha_i(n) \,\leqa\, \log(n)$ for $i=1,2$, then 
\[\beta_1\,\leqi\,\beta_2 \en \mbox{ iff } \en \alpha_1\,\leqpow\,\alpha_2 \quad \mbox{ and } \quad 
\beta_1\,\lei\,\beta_2 \en \mbox{ iff } \en \alpha_1\,\lepow\,\alpha_2,\]
and from $\,2\,\leqa\, \alpha_1\,\llpow\,\alpha_2\,$ it follows that $\,\beta_1\,\lei\,\beta_2\,$.
\end{lemma}

\begin{proof}
To show the inequality $\It(\beta)\,\leqa\,\{ n^{(\alpha(n))^l}: \; l\in\N\}$, we inductively prove that for all $k\in\Nplus$ there exists an $l\in\N$ such that
\[\beta\ite{k}(n)\,\leqa\, n^{(\alpha(n))^l}.\]
This holds for $k=l=1$. If it holds for some $k$ and $l$, then
\bea
\beta\ite{k+1}(n)&=& \beta\circ\beta\ite{k}(n) \\
&\leqa& \beta( \,n^{(\alpha(n))^l}\,) \\
&=& ( n^{(\alpha(n))^l})^{(\,\alpha( n^{(\alpha(n))^l})\,)       }\\
&\leqa& ( n^{(\alpha(n))^l})^{(\,\alpha( n^{(\log(n))^l})\,)       }\\
&\leqa&  ( n^{(\alpha(n))^l})^{(\,(\alpha( n))^{l'}\,)       }\\
&=&   n^{(\,\alpha(n))^{l+l'}},
\eea
for some $l'\in\Nplus$.

For the converse inequality, we use that  $n\leqa\beta\ite{k}(n)$ and  show that 
\[  n^{(\alpha(n))^k}\,\leqa\, \beta\ite{k}(n).\]
This holds for $k=1$, and if it holds for some $k$, then 
\bea
 n^{(\alpha(n))^{k+1}} & \leqa & ( \beta\ite{k}(n))^{\alpha(n)} \\
 & \leqa & ( \beta\ite{k}(n))^{\alpha(\beta\ite{k}(n))} \\
& = & \beta( \,\beta\ite{k}(n)\,) \\
& = & \beta\ite{k+1}(n).
\eea

Like in the proof of Lemma 13, the second assertion follows immediately,
and Lemma12 yields the last one.
\qed
\end{proof}

Next we show an analogue of Lemma 14.

\begin{lemma}
Let $\alpha_1$ and $\alpha_2$ be nondecreasing, e-consistent functions with
 $1\,\leqa \,\alpha_i(n)\leqa\log(n)$ for $i=1,2$.
Then the function $\alpha$ defined by
\[\alpha(n) = 2^{\msqrti(\,(\mlogi\,\circ\,\alpha_1(n))\cdot(\mlogi\,\circ \,\alpha_2(n))\,)  }\]
fulfills these properties, too.
\end{lemma}
 \begin{proof}

Analogously to the proof of Lemma 14, we only have to prove that $\alpha$ is e-consistent, i.e., that to any $k\in\N$ there is an $l\in\N$ with 
\[ 2^{\msqrti(\,(\mlogi\,\circ\,\alpha_1(n^{(\log(n))^k}))\cdot(\mlogi\,\circ \,\alpha_2(n^{(\log(n))^k}))\,)  } \,
\leqa \,( 2^{\msqrti(\,(\mlogi\,\circ\,\alpha_1(n))\cdot(\mlogi\,\circ \,\alpha_2(n))\,)  })\,^l. \]
This means that \\
\hspace*{0.9cm}
$ \msqrt(\,(\mlog\,\circ\,\alpha_1(n^{(\log(n))^k}))\cdot(\mlog\,\circ \,\alpha_2(n^{(\log(n))^k}))\,)$  \\
\hspace*{1cm}\hfill
$\leqa \,   l\cdot \msqrt(\,(\mlog\,\circ\,\alpha_1(n))\cdot(\mlog\,\circ \,\alpha_2(n))\,)$ , \\
and this follows if to any $k$ there is an $m$ such that 
\[(\mlog\,\circ\,\alpha_1(n^{(\log(n))^k}))\cdot(\mlog\,\circ \,\alpha_2(n^{(\log(n))^k}))
\,\leqa \,  m \cdot(\mlog\,\circ\,\alpha_1(n))\cdot(\mlog\,\circ \,\alpha_2(n)).\] 
Since $\alpha_1$ and $\alpha_2$ are e-consistent, there are numbers $l_1,l_2\in\N$ such that
\bea
(\mlog\,\circ\,\alpha_1(n^{(\log(n))^k}))\cdot(\mlog\,\circ \,\alpha_2(n^{(\log(n))^k})) &&
\eea
\vspace*{-1.25cm}

\bea \hspace*{6.4cm}
& \leqa & (\mlog\,\circ\,(\,(\alpha_1(n))^{l_1}\,)\,\cdot\, (\mlog\,\circ\,(\,(\alpha_2(n))^{l_2}\,)\\
& \leqa & l_1\cdot(\mlog\,\circ\,\alpha_1(n)\,)\,\cdot\,  l_2\cdot (\mlog\,\circ\,\alpha_2(n)\,)\\
& = & l_1\cdot l_2\cdot(\mlog\,\circ\,\alpha_1(n)\,)\cdot (\mlog\,\circ\,\alpha_2(n)\,).
\eea
So the assertion has been shown.
\qed
\end{proof}

By means of Lemmas 15, 17 and 18, a dense universe of bound above the polynomial ones can be constructed 
analogously to the proof of Theorem 1.

\begin{definition}[Bounds of Type 2]
A time bound $\beta$ is said to be of Type 2 \en iff
\[\beta(n)= n^{\alpha(n)}\]
with a nondecreasing function $\alpha$ such that $2\,\leqa\,\alpha(n)\,\leqa\,\log(n)$ and $\alpha$ is e-consistent and constructible in linear time.
\end{definition}
It should be noticed that all bounds of Type 2 are time-constructible.
The polynomials $\beta(n)=n^k$ with $k\geq 2$ represent the (with respect to $\leqi$) smallest bounds of Type 2.
The largest one is given by the smallest quasipolynomial
$\beta(n)= n^{\log(n)}$ which is it-equivalent to $n^{(\log(n))^k} \eqi 2^{(\log(n))^{k+1}}$ for all $k\in\Nplus$.

\begin{theorem}
There is a tame set $U_2$ consisting of time bounds of Type 2 \en such that 
\[\Bpol  \, \cup \,\Bqpol \; \emi \; U_2\]
and $U_2$ is densely linearly ordered by $\,\leqi$.
 Moreover, if we have $\beta_i(n)= n^{\alpha_i(n)}$, for $i=1,2$ and 
$\beta_1,\beta_2\in U_2$, then it holds:
$\;\beta_1\lei\beta_2\:$ iff $\:\alpha_1\llpow\alpha_2$.
\end{theorem}

\begin{proof}
Analogously to the proof of Theorem 1, let
$\: U_2=  \bigcup\nolimits_{i\in\Ni} W_i\,,\:$
where we now start with the two-element set
$W_0= \{ n^2,  n^{\log(n)} \}$ and put
\bea
W_{j+1}& =& W_j \,\cup\,\{\beta: \mbox{ there are functions $\beta_1, \beta_2\in W_j$ such that }  \beta_1\,\lei\,\beta_2,\\
&& \hspace*{1.97cm} \mbox{but there is no $\beta_3\in W_j$ with $\beta_1\lei\beta_3\lei\beta_2$},\\
&& \hspace*{1.97cm} \mbox{and }\:
\beta(n)=n^{\alpha(n)}\: \mbox{ with }
\alpha(n) = 2^{\msqrti(\,(\mlogi\,\circ\,\alpha_1(n))\cdot(\mlogi\,\circ \,\alpha_2(n))\,)  }, \\
 && \hspace*{7.2cm} \mbox{where $\beta_i(n)=n^{\alpha_i(n)}\,$ for $i=1,2$} \,\},
\eea
for all $j\in\N$. 

By  Lemmas 15, 17 and 18, one sees that the assertions of Theorem 2 are valid.
\qed
\end{proof}

\begin{corollary}
There is a tame set $U_0$, which consists of time bounds of Type  1 and 2 only, \linebreak
 such that 
\[\Blin \,\cup\, \Blogs \,\cup\, \bigcup\nolimits_{m\in\Nplusi} \! \Bqmlin{m} \, \cup \,\Bpol\, \cup \,\Bqpol \; \emi \; U_0\]
and $U_0$ is densely linearly ordered by $\,\leqi$.
The set $\It(U_0)$ is it-closed, tame and consists of time-constructible bounds only.
\end{corollary}

\begin{proof}
This immediately  follows by putting $U_0=U_1\cup U_2$ with the sets $U_1$ and $U_2$ taken from the proofs of Theorems 1 and 2, respectively. 
Then it holds $\It(U_0)= \It(U_1)\cup\It(U_2)$, and the tameness holds by construction.
\qed
\end{proof}

It should be noticed that the sets $W_0$ from the proofs of Theorems 1 and 2 could arbitrarily enlarged by further bounds of Type 1 and 2, respectively, as long as they remain linearly ordered  by $\,\leqi$ and yield tame iteration sets.

Unfortunately, the proof of (the first part of) Lemma 17 is based on the supposition that $\alpha(n)\leqa \log(n)$. So the quasipolynomials are the maximal bounds we can reach by our technique. It remains an open problem to fill the gap between the quasipolynomial $\beta(n)= n^{\log(n)}$ and the exponential function 
$\beta_{\mbox{\footnotesize exp}}(n)= 2^n$
by another set of bounds densely ordered with respect to $\leqi$.

\section{Closures and cuts over certain universes}

In order to get an impression of the variety of regular sets,
we first study the systems of regular subsets over special universes.
For an arbitrary universe $U$ of time bounds, let
$\Reg(U)$ denote the related system:
\[ \Reg(U) =\{B: B\subseteq U \mbox{ and $B$ is regular}\,\}.\]
Since the relation $\leq$ (meaning $\leqa$ or $\leqo$) is  reflexive and transitive but not antisymmetric, it is natural to consider the canonical factorization of $\Reg(U)$, i.e.,
\[\ov{\Reg(U)}=\: \parbox{1.55cm}{$\Reg(U)_{{\mbox{\normalsize $\!/\!$}}_{\mbox{\normalsize $\equiv$}}}$}\: =\{ [B]: \: B\in\Reg(U)\,\}.\]
It consists of the {\it confinality classes}\/ over $U$, 
\[[B]=\{ B': \, B'\in\Reg(U) \mbox{ and } B'\equiv B\,\}.\]
The relation $\leq$  is canonically transferred from $\Reg(U)$ to $\ov{\Reg(U)}$ by
\[ [B_1] \leq [B_2] \en \mbox{ iff } \en B_1\leq B_2 \qquad \mbox{(for $B_1,B_2\in\Reg(U)\,$)}.\]
Obviously, the validity of  $ [B_1] \leq [B_2] $ does not depend on the choice of representatives $B_1$ and $B_2$ of the confinality classes, and $\leq$ is a reflexive, transitive and antisymmetric relation in $\ov{\Reg(U)}$.
So the ordered set $(\,\ov{\Reg(U)}, \leq\,)$ represents a more concise view to the confinality between regular subsets of the universe $U$. Nevertheless, we shall often prefer to work directly with the  representatives $B\in\Reg(U)$ instead of their 
classes $[B]\in\ov{\Reg(U)}$.

For an arbitrary binary relation $\leq$ in an arbitrary universe $U$, a subset $S\subseteq U$ is said to be (downwards) {\it closed}\/  under $\leq$ if it holds $x\in S$ whenever $x\leq y$ for some $y\in S$. It turns out that, for regular subsets $B$ of an arbitrary universe $U$ of time bounds, the three concepts of closedness with respect to $\leqa$, $\leqo$ and $\leqi$, respectively, coincide.
So the {\it closure}\/ of $B$ within $U$ can be defined as 
\[ \Cl_U(B) = \{\beta\in U:\, \mbox{there is a $\beta'\in B$ with $\beta\leqi\beta'$}\,\}.\]

\begin{lemma}
Let $U$ be a universe of time bounds and $B\subseteq U$ a regular subset. Then $B$ is closed under $\leqa$ iff it is closed under $\leqo$ iff it is closed under $\leqi$, and this holds iff $B=\Cl_U(B)$. 
\end{lemma}

\begin{proof}
First let $B$ be closed under $\leqa$ and $\beta\leqo\beta'\in B$ for some $\beta\in U$. Then $\beta\leqa c\cdot \beta'$ for some $c\in\Nplus$. Since $B$ is regular, there is a $\beta''\in B$ with $ c\cdot \beta' \leqa \beta''$, hence $\beta\leqa \beta''$ and
 $\beta\in B$.

Now let $B$  be closed under $\leqo$ and $\beta\leqi\beta'\in B$. Then $\beta\leqa \beta'\ite{m}$ for some $m\in\Nplus$.
Since $B$ is regular, it follows $\beta'\ite{m}\leqa \beta''$ for some $\beta''\in B$, hence $\beta\leqa\beta''$, thus $\beta\leqo\beta''$ and $\beta \in B$.

Finally, if $B$ is closed under $\leqi$ and $\beta\leqa\beta'\in B$, then $\beta\leqi \beta'$, hence $\beta\in B$.

It always holds $B\subseteq \Cl_U(B)$, and if $B$ is closed under $\leqi$, we also have $\Cl_U(B)\subseteq B$.
On the other hand, $\Cl_U(B)=B$ implies that $B$ is closed under $\leqi$.
\qed
\end{proof}

By Lemma 19, the closure of a regular set $B\subseteq U$ could equivalently be defined by means of $\leqo$ or $\leqa$, e.g.,
\[ \Cl_U(B) = \{\beta\in U:\, \mbox{there is a $\beta'\in B$ with $\beta\leqa\beta'$}\,\},\]
and it is the (with respect to inclusion) smallest superset of $B$ which is closed (with respect to $\leqa$, $\leqo$ or $\leqi$).

Of course, not every regular set of bounds is closed. However, the closed regular sets yield a (with respect to confinality) complete
system of representatives of all regular sets over a given universe of bounds.

\begin{lemma}
Let $U$ be a universe of time bounds. For any regular subset $B\subseteq U$, the closure $\Cl_U(B)$  is regular, too, and it holds
\[B\equiv\Cl_U(B).\]
\end{lemma} 

\begin{proof}
One easily shows the conditions i) and ii) of regularity for $\Cl_U(B)$. Since $B\subseteq \Cl_U(B)$, we have $B\leqa \Cl_U(B)$.
$\Cl_U(B)\leqa B$ follows by the characterization of $\Cl_U(B)$ given after the proof of the preceding lemma.
\qed
\end{proof}

\begin{lemma}
For regular sets $B_1,B_2\subseteq U$, it holds $B_1\leq B_2$ iff $\Cl_U(B_1)\subseteq \Cl_U(B_2)$. Thus,
\[B_1\equiv B_2 \en \mbox{ iff } \en \Cl_U(B_1)= \Cl_U(B_2).\]
\end{lemma}

\begin{proof}
This again is a consequence of the above representation of $\Cl_U(B)$ for regular $B$ by means of $\leqa$. Indeed, it immediately shows that $B_1\leq B_2$ implies  $\Cl_U(B_1)\subseteq \Cl_U(B_2)$. If the latter inclusion holds and 
$\beta_1\in B_1\subseteq \Cl_U(B_1)\subseteq \Cl_U(B_2)$, then $\beta_1\in\Cl_U(B_2)$. Thus, there is a $\beta_2\in B_2$ such that $\beta_1\leqa\beta_2$. This means that   $B_1\leq B_2$.
\qed
\end{proof}

Whereas not every set which is closed under $\leqa$ or $\leqo$ has to be regular, the sets which contain properly linear bounds and are closed under $\leqi$ are indeed regular if the universe fulfils some natural requirements:

\begin{lemma}
Let $U$ be an it-complete and it-tame universe of time-constructible bounds.
Then every subset $B\subseteq U$, which is closed under $\leqi$  and contains a bound $\beta_0\in B$ such that $2n\leqa \beta_0$, is regular.
\end{lemma}

\begin{proof}
Condition i) of regularity is fulfilled for all $B\subseteq U$, since all elements of $U$ are time-constructible.
To show condition ii), let $\beta,\beta'\in B$.
\\
Case 1:  $\beta\leqi\beta'$, say $\beta\leqa\beta'\ite{m}$ for some $m\in\Nplus$. Then 
$\beta +\beta'\circ\beta \leqa \beta'+ \beta'\ite{m+1} \leq  \beta'\ite{m+1}+ \beta'\ite{m+1}= 2\cdot  \beta'\ite{m+1}$.
\\
For $\beta'\leqi \beta_0$, say $\beta'\leqa \beta_0\ite{l}$, we have
$\beta +\beta'\circ\beta \leqa 2\cdot  \beta'\ite{m+1}\leqa 2\cdot  \beta_0\ite{l\cdot(m+1)}\leqa  \beta_0\ite{l\cdot(m+1)+1}$.
Since $U$ is it-complete, there is a subset $\widetilde{B}\subseteq U$ such that 
$\It(\beta_0)\eqa \widetilde{B}$. This means that $ \beta_0\ite{l\cdot(m+1)+1}\leqa \widetilde{\beta}\leqa$ $\beta_0\ite{k}$
for some $\widetilde{\beta}\in\widetilde{B}$ and $k\in\Nplus$. It follows that $\widetilde{\beta}\eqi \beta_0\in B$,
hence $\widetilde{\beta}\in B$ due to the closedness of $B$ under $\leqi$, and we have $\beta +\beta'\circ\beta \leqa    \widetilde{\beta}$.
\\
For $\beta_0\leqi \beta'$, say $\beta_0\leqi \beta'\ite{l}$, it follows 
$\beta +\beta'\circ\beta \leqa 2\cdot  \beta'\ite{m+1}\leqa \beta'\ite{l\cdot(m+1)}$.
Since $U$ is it-complete, analogously to the above conclusion, one finds a  $\widetilde{\beta}\in B$ such that 
$\beta +\beta'\circ\beta \leqa    \widetilde{\beta}$.
\\
Case 2:  $\beta'\leqi\beta$, say $\beta'\leqa\beta\ite{m}$ for some $m\in\Nplus$.
Then 
$\beta +\beta'\circ\beta \leqa \beta\ite{m+1}+ \beta\ite{m+1} = 2\cdot  \beta\ite{m+1}$.
\\
For $\beta\leqi\beta_0$, it follows that $2\cdot  \beta\ite{m+1}\leqa \beta_0\ite{l\cdot(m+1)+1}$,
if $\beta_0\leqi\beta$, we have $2\cdot  \beta\ite{m+1}\leqa \beta\ite{l\cdot(m+1)}$, with a suitable $l\in\Nplus$ in both cases.
Due to the it-completeness of $U$ and the closedness of $B$ under $\leqi$, we again find a
$\widetilde{\beta}\in B$ such that $\beta +\beta'\circ\beta \leqa    \widetilde{\beta}$.
\\
So the regularity of $B$ has been shown.
\qed
\end{proof}

Now we are able to characterize the order structure of regular sets over certain universes $U$, i.e., the ordered set
 $(\,\ov{\Reg(U)},\leq\,)$, by means of the structure of non-empty $\leqi$-closed subsets ordered by the set-theoretic inclusion.
For some more background on the theory of (linearly) ordered sets, the reader is referred to \cite{Ro}. Details of topological notions
and results can be found in \cite{En,HY}.

Given a universe $U$ of bounds, we put
\[\Clit(U) =\{ B:\:\emptyset\not= B\subseteq U \mbox{ and $B$ is closed  under $\leqi$}\,\}.\]

\begin{proposition}
Let $U$ be an it-complete and it-tame universe of time-constructible bounds $\beta$ such that always $2n\leqa \beta$.
Then the mapping $\varphi$ defined by $\,\varphi([B])= \Cl(B)$, for all $[B]\in\ov{\Reg(U)}$, is an order-isomorphism
between the ordered sets $(\,\ov{\Reg(U)},\leq\,)$ and  $(\,\Clit(U),\subseteq\,)$.
\end{proposition}

\begin{proof}
By Lemma 21, $\varphi([B])$ is independent on the choice of a representative  of the class $[B]$, and $\varphi$ is an injective 
order-homomorphism of $(\,\ov{\Reg(U)},\leq\,)$ to  $(\,\Clit(U),\subseteq\,)$.
 Lemma 22 shows that for every $B\in\Clit(U)$ it holds $B\in\Reg(U)$. Then we have 
$\varphi([B])=\Cl(B)=B$, hence $\varphi$ is surjective.
\qed
\end{proof}

Now let $U_0$ be a tame universe of time bounds of Type 1 and 2 which is densely linearly ordered by $\leqi$, see Corollary 3.
The iterative closure $U=\It(U_0)$ fulfills all the suppositions of Proposition 4. Thus, the ordered sets 
$(\ov{\Reg(U)},\leq)$ and $(\Clit(U),\subseteq)$ are order-isomorphic.
By Lemma 7, they are even linearly ordered sets.

The closed sets $B\in\Clit(U)$ are characterized by their elements belonging to $U_0$. More precisely, for each $B\in\Clit(U)$ we have 
\[ B = \bigcup \{\It(\beta):\: \beta\in B\cap U_0\,\},\]
and for $B_1,B_2\in \Clit(U)$ it holds
\[B_1\subseteq B_2 \quad \mbox{ iff } \quad B_1\cap U_0 \subseteq B_2\cap U_0.\]

\begin{lemma} Let $U_0$ be a universe of time bounds according to Corollary 3,
and let $U=\It(U_0)$.
Then the mapping $\varphi$ defined by $\varphi(B) = B\cap U_0$,  for any $B\in\Clit(U)$,
is an order-isomorphism of $(\Clit(U),\subseteq)$ onto $(\Clit(U_0),\subseteq)$.
\end{lemma}

\begin{proof}
One easily sees that $B\cap U_0 \in \Clit(U_0)$ if $B\in\Clit(\It(U_0))$. It remains to show that $\varphi$ is surjective.
This however follows, since for any $B_0\in\Clit(U_0)$ we have $\It(B_0)\in\Clit(U)$ and $\varphi(\It(B_0))=\It(B_0)\cap U_0=B_0$.
\qed 
\end{proof}

So we have shown that the structure of confinality classes of regular sets over $U=\It(U_0)$, i.e., the ordered set
$(\,\ov{\Reg(U)},\leq\,)$ is isomorphic to $(\Clit(U_0),\subseteq)$. Now we are going to explain that the latter structure has a
well-known order theoretic characterization which is closely related to the Cantor discontinuum.

We start with the relationship between the nontrivial closed subsets and the cuts over linearly ordered sets.
For an arbitrary linearly ordered set $(U,\leq)$, let 
\[ \Csets(U)=\{ S:\: \emptyset \subset S\subset U
\mbox{ and $S$ is closed under $\leq$}\,\}.\]
It is easily seen that this system is linearly ordered by the inclusion, i.e., $(\Csets(U),\subseteq)$ is a linearly ordered set.
\\
By a {\it cut}\/ over  $(U,\leq)$, one understands a pair 
\[C=(C_1,C_2)\]
such that $C_1\in \Csets(U)$, $C_1\cap C_2=\emptyset$, and $C_1\cup C_2=U$.
$C_1$ is called the {\it lower segment}\/, whereas $C_2$ is the {\it upper segment}\/ of the cut $C$.
The system of all cuts,
\[ \Cuts(U)=\{C:\: C \mbox{ is a cut over  $(U,\leq)$}\},\]
is canonically ordered by the inclusion between the lower segments. More precisely, let
\[C\, \incl\, C' \quad \mbox{ iff }\quad C_1\subseteq C_1'\, , \quad \mbox{ for cuts $C=(C_1,C_2)$ and $C'=(C_1',C_2')$.}\]
Again, it is easily seen that $(\Cuts(U),\incl)$ is a linearly ordered set. Moreover, the mapping 
$\fct{\varphi}{\Csets(U)}{\Cuts(U)}$ defined by 
\[ \varphi(S)= (\,S\,,\,U\setminus S\,), \quad \mbox{ for all $S\in\Csets(U)$},\]
is an order-isomorphism between $(\Csets(U),\subseteq)$ and  $(\Cuts(U),\incl)$.
So we have 
\begin{lemma}
For any linearly ordered set $(U,\leq)$, there is an order-isomorphism of  $(\Csets(U),\subseteq)$ onto  $(\Cuts(U),\incl)$.
\qed
\end{lemma}

The notion of cuts over ordered sets is due to R. Dedekind \cite{De} who introduced the real numbers by means of cuts
over the naturally ordered set of rational numbers $(\Q,\leq)$. In fact, for 
\[ \Cuts^0(\Q)=\{ (C_1,C_2)\in\Cuts(\Q):\: \mbox{ $C_1$ does not contain a greatest element}\},\]
it turns out that $(\Cuts^0(\Q),\incl)$ is order-isomorphic to the ordered set of reals $(\R,\leq)$, how ever this might be introduced 
or generated.
Moreover, the arithmetical operations of addition and multiplication can be generalized from $\Q$ to $\R$ by means of 
Dedekind cuts. The restriction to $\Cuts^0(\Q)$, instead of taking $\Cuts(\Q)$, is necessary, since
each rational number $r\in\Q$ determines exactly two different cuts, namely 
\[ (\{x\in\Q:\, x<r\},\{x\in\Q:\, r\leq x\})\: \incl \: (\{x\in\Q:\, x\leq r\},\{x\in\Q:\, r< x\}).\]
Only the first one of these two cuts belongs to $\Cuts^0(\Q)$.

According to Lemma 24, we are interested in characterizing the structure of all closed sets and cuts, respectively,
over certain  universes $U$.
Notice a slight difference between these structures: the whole universe $U$ has been excluded in defining $\Csets(U)$, in order to obtain an isomorphism between  $(\Csets(U),\subseteq)$ and  $(\Cuts(U),\incl)$. We shall have to take this into account in applying our results later.

Now we employ a well-known result by Cantor, see \cite{Ro}.

\begin{lemma}[Cantor]
Any two densely linearly ordered sets, $(U_1,\leq_1)$ and $(U_2,\leq_2)$ with countably infinite universes
$U_1$ and $U_2$ are order-isomorphic to each other if both of them have or don't have least and greatest elements, respectively.
\qed
\end{lemma}

This yields a nice characterization of the order structure of cuts over related universes by means of the Cantor discontinuum.
Under the {\it Cantor discontinuum}\/, we here understand the ordered set
\[(\,\D\,,\,\leqx\,), \quad \mbox{ where } \en \D= \{0,1\}^{\Ni},\]
this is the set of all {\it binary sequences}\/ $\fct{f}{\N}{\{0,1\}}$; $\leqx$ denotes the {\it lexicographic ordering}\/, i.e., 
$f_1\leqx f_2\,$ iff $\,f_1=f_2$ or $\, f_1\not=f_2$ and $f_1(n_0)< f_2(n_0)$, where $n_0=\min\{n:\, f_1(n)\not= f_2(n)\}$.
It is well-known that $(\D,\leqx)$ is order-isomorphic to the {\it Cantor ternary set}\/, which is a subset of the real closed interval $[0,1]$ ordered by the natural order between real numbers.

Even if we did not found the following result explicitly in literature, it is surely folklore.

\begin{proposition}
Let $(U,\leq)$ be a densely linearly ordered set with a countably infinite universe $U$ having a least but no greatest element.
Then the structure of cuts over $U$, i.e., the ordered set $(\Cuts(U),\incl)$, is order-isomorphic to the Cantor discontinuum 
$(\D,\leqx)$.
\end{proposition}

\begin{proof}
It can be used that an arbitrary isomorphism 
$\fct{\varphi}{U}{U'}$ between two linearly ordered sets $(U,\leq)$ and $(U',\leq')$ immediately yields an order-isomorphism 
$\fct{\ov{\varphi}}{\Cuts(U)}{\Cuts(U')}$ between $(\Cuts(U),\incl)$ and $(\Cuts(U'),\incl)$. Indeed, this is fulfilled by putting
\[ \ov{\varphi}(C_1,C_2)=(\{\varphi(x):\, x\in C_1\},\{\varphi(x):\, x\in C_2\}), \quad 
\mbox{ for any cut $(C_1,C_2)\in \Cuts(U)$}.\]
So, due to Lemma 25, we can restrict ourselves to a special representative of a densely linearly ordered set with related properties.
Let $\B$ denote the set of binary words terminating with the letter $1$, enriched by the {\it empty word}\/ $\Lambda$,
\[\B=\{w\in \{0,1\}^*:\: w=\Lambda \mbox{ or the last letter of $w$ is $1$}\}.\]
By $\wlex$, we denote the lexicographic ordering in $\B$, i.e.,
\[w_1\,\wlex\, w_2 \mbox{ iff } \mbox{ \begin{minipage}[t]{12cm}
                                              $w_1$ is an initial piece of $w_2$ or \\
                                             if $w_0$ is the largest common initial piece of $w_1$ and $w_2$,\\
then $w_00$ is an initial piece of $w_1$ and  $w_01$ is an initial piece of $w_2$.
\end{minipage} } \]
Obviously, $(\B,\wlex)$ is a countably infinite and densely linearly ordered set with the least element $\Lambda$ but without a greatest element. For example, the words $11\ldots 1$ form an unbounded subset of $\B$.
If we put $\,\val(\Lambda)=0\,$ and 
\[\val(x_l\,\ldots\,x_1)= \sum\nolimits_{i=1}^l x_i\cdot 2^{-i}\,, \quad
                         \mbox{for $l\in\Nplus$ and $x_1,\ldots,x_l\in\{0,1\}$}, \]
then $\,\{\val(w):\, w\in\B\}\,$ is the set of so-called {\it dyadic rational numbers}\/ within the half-open interval $[0,1)$,
 and it holds
\[w_1\wlex w_2 \quad \mbox{ iff } \quad \val(w_1)\leq \val(w_2).\]

Let $C=(C_1,C_2)\in\Cuts(\B)$. There is exactly one sequence $\,f_C\in \D\,$ such that
\begin{itemize}
\item[ i)] all $w_1\in C_1$ are initial pieces of $f_C$, and
\item[ii)] no $w_2\in C_2$ is an initial piece of $f_C$.
\end{itemize}
It is left to the reader to define $f_C$ explicitly in dependence on the segments $C_1$ and $C_2$ of the cut $C$.
For example, $f_C= 0^{\infty}$ for $C_1= \{\Lambda\}$;
if $w1= \min(C_2)$ for some $w\in\{0,1\}^*$, then $f_C=w01^{\infty}$;
if $w1= \max(C_1)$ for $w\in\{0,1\}^*$, then $f_C=w10^{\infty}$.
(Here we concatenate words with sequences in the straightforward way, and $x^{\infty}$ denotes the stationary sequence of 
a letter $x\in\{0,1\}$.)

Finally, it is easily seen that for any two cuts, $C$ and $C'$ from $\Cuts(\B)$, it holds $C\incl C'$ iff $f_C\leqx f_{C'}$.
Thus, the mapping $\varphi(C)= f_C$ defines an order-isomorphism of $(\Cuts(\B),\incl)$ into $(\D,\leqx)$.
Moreover, $\varphi$ is bijective.
\qed 
\end{proof}

For $U=\It(U_0)$, where the universe $U_0$ is taken according to Corollary 3, let 
\bea
 \ov{\Reg(U)}^0 &=& \ov{Reg(U)} \setminus \{ [U_0]\} \quad \mbox{ and}\\
\Clit^0(U_0) &=& \Clit(U_0) \setminus \{U_0\}.
\eea
Then we have the following order-isomorphisms:
\[( \ov{\Reg(U)}^0,\leq) \cong (\Clit^0(U_0),\subseteq) \cong (\Cuts(U_0),\incl) \cong (\D,\leqx).\]
The linearly ordered sets $( \ov{\Reg(U)},\leq)\cong(\Clit(U_0),\subseteq)$,
where $\leq$ means $\leqa$ or $\leqo$ for representatives of the related classes,
 are order-isomorphic to the extension of the Cantor discontinuum by just one isolated element which is greater than all the others, i.e., it has to be an immediate successor of the sequence $1^{\infty}\in\D$.

In particular, there is an order-isomorphic embedding $\varphi_1$ of the Cantor discontinuum $(\D,\leqx)$ into
 $( \ov{\Reg(U)},\leq)$ such that $\varphi_1(0^{\infty})= [\It(2n)]$, 
$\varphi_1(1^{\infty})= [\{\beta\in U:\, \beta\lei n^{\log(n)}\}] = [\{\beta\in U_0:\, \beta\lei n^{\log(n)}\}]$
 and $\ran(\varphi_1) =  \ov{\Reg(U)}^0$.
On the other hand, we also have an order-isomorphic embedding $\varphi_2$ of the real closed interval $[0,1]$ into the Cantor discontinuum. This can be defined by means of the binary expansions of the reals $r\in[0,1]$: If 
$r=\sum_{i=0}^{\infty} x_i\cdot 2^{-(i+1)}$ with digits $x_i\in\{0,1\}$ such that either $r=1$ or $x_i=0$ for infinitely many $i\in\N$, then put
 \[ \varphi_2(r) =f_r\in\D, \quad \mbox{ where } f_r(i)=x_i \mbox{ for all } i\in\N.\]
It holds
$\varphi_2(0)=0^{\infty}$, $\varphi_2(1)=1^{\infty}$, and $\varphi_2(r_1)\leqx \varphi_2(r_2)$ iff $r_1\leq r_2$, for all
$r_1,r_2\in[0,1]$.
Thus, 
$\varphi_1\circ\varphi_2$ is an order-isomorphism of $[0,1]$ onto $\ov{\Reg(U)}^0$ which maps $0$ to the confinality class
$[\It(2n)]$ and $1$ to $[\{\beta\in U:\, \beta\lei n^{\log(n)}\}]= [\{\beta\in U_0:\, \beta\lei n^{\log(n)}\}]$. Also, it is well-known that the whole continuum $\R$ is order-isomorphic to the open interval $(0,1)$. Hence $\R$ is order-isomorphically embeddable into
  $( \ov{\Reg(U)},\leq)$.

Finally, it should be noticed that the notion of regular set is an absolute one. This means that it does not depend on the universe of bounds just considered. Thus, with respect to the maximal universe $U'\imax$ of all time-constructible bounds,
 see the end of Section 2,  we conclude that both
 $(\D,\leqx)$ and $([0,1],\leq)$ as well as  $(\R,\leq)$ are order-isomorphically embeddable into the ordered set 
of all confinality classes $(\ov{\Reg(U'\imax)},\leq) \cong (\ov{\Reg(U\imax)},\leq)$. Of course, the  relation $\leq$
 is not a linear ordering in $\ov{\Reg(U'\imax)}$ or $\ov{\Reg(U\imax)}$.
 
The following theorem summarizes the results just obtained.

\begin{theorem}
For each of the linearly ordered sets  $(\D,\leqx), \, ([0,1],\leq)$ and  $(\R,\leq)$ there is an an order-isomorphic embedding into
the ordered set $(\ov{\Reg(U'\imax)},\leq)\cong (\ov{\Reg(U\imax)},\leq)$ of all confinality classes of regular sets of bounds.\\
In particular, an order-isomorphic embedding $\varphi$  of  $([0,1],\leq)$ can be defined in such a way that 
$\varphi(0)\equiv \Blin$, $\varphi(1)\equiv [\{\beta\in U_0:\, \beta\lei n^{\log(n)}\}]$, and 
$\ran(\varphi)$ consists of all classes confinal to elements of $\ov{\Reg(U)}^0$, where $U_0$ is the universe according to Corollary 3, and $U=\It(U_0)$.
\qed
\end{theorem}

\section{On o-regular sets, oracle hierarchies and alternation}

Now we shall show that the iteration sets $\It(\beta)$ of time bounds $\beta$  of Type 1 or 2 are even o-regular. 
The notion of o-regularity of a set of bounds was introduced in \cite{He1}. It guarantees a certain robustness of the complexity classes defined by means of oracle machines of related time bounds. We shall see that the hierarchies build by such classes 
coincide with those obtained by time-bounded alternating Turing machines.

\begin{definition}
A regular set $B$ of time bounds is said to be {\it o-regular}\/ iff
for each $\beta\in B$, there is some $\beta'\in B$ such that for almost all $n\in\Nplus$ the following holds:
\[ \mbox{if \en$\sum_{i=1}^l m_i \leq \beta(n)\,$ for some $l,m_1,\dots ,m_l\in\Nplus$,
then  \en $\sum_{i=1}^l \beta(m_i) \leq \beta'(n)\,$.}\]
\end{definition}

This is a slightly generalized version of the notion of o-regularity from \cite{He1},
which still allows to show all the results proved there, simply by straightforward adaptations of the given proofs.
The present version has first been used in \cite{Ko}. Analogously to Lemma 6, it is proved that any set $B$ of time-constructible bounds which is confinal to an o-regular set is o-regular, too.
The following result shows that Theorems 1 and 2 yield a large variety of o-regular sets.

\begin{proposition}
If $\beta$ is a time bound of Type 1 or 2, then the iteration set $\:\It(\beta)$ is o-regular.
\end{proposition}

\begin{proof}
First let $\beta$ be of Type 1, i.e.,
\[\beta(n) = n\cdot \alpha(n) \quad
\mbox{with a nondecreasing, f-consistent function $\alpha$ satisfying $2\leqa\alpha(n)\leqa n$.} \]
The regularity of $\It(\beta)$ follows by Lemma 9. To show the o-regularity, let $\ov{\beta}\in\It(\beta)$.
By Lemma 13, we have $\ov{\beta}(n)\leqa n\cdot \alpha^{\overline{m}}(n)$, with a suitable $\overline{m}\in\Nplus$. Thus,
for almost all $n$, from  $\sum_{i=1}^l m_i \leq \ov{\beta}(n)\,$ with $l,m_1,\dots ,m_l\in\Nplus\,$ it follows that
$l\leq \ov{\beta}(n)$ and 
\[ m_i \leq \sum\nolimits_{i=1}^l m_i \leq n\cdot \alpha^{\overline{m}}(n) \leq n^{\overline{m}+1}.\]
So there are suitable $c,\widetilde{m},\widehat{m}\in\N$ such that for almost all $n\in\N$:
\bea
\sum\nolimits_{i=1}^l \ov{\beta}( m_i )&\leq& 
c\cdot \ov{\beta}(n) + \sum\nolimits_{i=1}^l ( m_i\cdot\alpha^{\overline{m}}( m_i )\,) \\
&\leq&
 c\cdot \alpha^{\overline{m}}( n ) + \left( \sum\nolimits_{i=1}^l m_i \right)\cdot  \alpha^{\overline{m}}( n^{\overline{m}+1} )\\
&\leq& n \cdot \alpha^{\overline{m}}( n ) +  n \cdot  \alpha^{\overline{m}}( n) \cdot  \alpha^{\widetilde{m}}( n)\\
&\leq& n \cdot \alpha^{\overline{m}}( 1+  \alpha^{\widetilde{m}}( n)\,)\\
&\leq& n \cdot \alpha^{\overline{m}}(\alpha^{\widetilde{m}+1}( n)\,)\\
&\leq& n \cdot \alpha^{\overline{m}}( n^{\widetilde{m}+1}\,)\\
&\leq& n \cdot \alpha^{\widehat{m}}( n)\\
&\leq& \beta'(n),
\eea
with a suitable $\beta'\in\It(\beta)$. 
(We remember that the exponents at function $\alpha$ refer to the exponentiation of the values but not to iterations of the function.)

Now let  $\beta$ be of Type 2, i.e.,
\[\beta(n) = n^{\alpha(n)} \quad 
\mbox{with a nondecreasing, e-consistent function $\alpha$ satisfying $2\leqa\alpha(n)\leqa \log(n)$.}\]
By Lemma 9, $\It(\beta)$ is regular.  To show the o-regularity, let $\ov{\beta}\in\It(\beta)$.
By Lemma 17, $\ov{\beta}(n)\leqa n^{ \alpha^{\overline{m}}(n)}$ with a suitable $\overline{m}\in\N$. 
From $\sum_{i=1}^l m_i \leq \ov{\beta}(n)\,$ with $l,m_1,\dots ,m_l\in\Nplus\,$ it follows 
$l\leq \ov{\beta}(n)$,
\[ m_i \leq \sum\nolimits_{i=1}^l m_i \leq n^{ \alpha^{\overline{m}}(n)} \en
\mbox{ and } \en
\alpha(m_i) \leq \alpha( n^{(\log(n))^{\overline{m}}}) \leq (\, \alpha(n)\,)^{\widetilde{m}}\]
for almost all $n\in\N$,
with a suitable $\widetilde{m}\in\N$.
So for almost all $n\in\N$ we have 
\bea
\sum\nolimits_{i=1}^l \ov{\beta}( m_i )
&\leq&  c\cdot \ov{\beta}(n) + \sum\nolimits_{i=1}^l (\, m_i^{\alpha^{\overline{m}}( m_i )}\,) \\
&\leq&  n^{ \alpha^{\overline{m}+1}( n )} +  n^{\alpha^{\overline{m}}( n)} \cdot  n^{\alpha^{\overline{m}}( n)\cdot
                                                                                                                                                       \alpha^{\overline{m}(m_i)}}\\
&\leq&  n^{ \alpha^{\overline{m}+1}( n )} +  n^{\alpha^{\overline{m}}( n)} \cdot  n^{\alpha^{\overline{m}}( n)\cdot
                                                                                                                       \alpha^{\widetilde{m}\cdot\overline{m}(n)}}\\
&\leq&  n^{\alpha^{\overline{m}+1}(n)}\cdot( \,1+ n^{ \alpha^{(\overline{m}+\widetilde{m}\cdot\overline{m})}(n)}\,)\\
&\leq&  n^{\alpha^{\overline{m}+1}(n)}\cdot (\,  n^{ \alpha^{(\overline{m}+\widetilde{m}\cdot\overline{m}+1)}(n)}\,)\\
&\leq&  n^{(\,\alpha^{\overline{m}+1}(n) + \alpha^{(\overline{m}+\widetilde{m}\cdot\overline{m}+1)}(n)\,)} \\
&\leq&  n^{(\,\alpha^{\overline{m}+1}(n)\, \cdot\, \alpha^{(\overline{m}+\widetilde{m}\cdot\overline{m}+1)}(n)\,)} \\
&\leq&  n^{(\,\alpha^{2\overline{m}+\widetilde{m}\cdot\overline{m}+2)}(n)\,)}\\
& \leq &  \beta'(n)\,,
\eea
with suitable $c,\widetilde{m}\in\N$ and some $\beta'\in\It(\beta)$.
\qed
\end{proof}

\begin{corollary}
Let $U$ be an it-complete and it-tame universe of bounds of Type 1 or 2.
Then every regular subset of $U$ is even o-regular.
\end{corollary}

\begin{proof}
Let $B$ be a regular subset of a universe $U$ fulfilling the supposition. By Lemmas 19, 20 and 22, we have 
\[ B= \Cl_U(B) \equiv  \bigcup\, \{\It(\beta):\: \beta\in B\}.\]
Thus, by Proposition 6, $B$ is the union of o-regular sets.
It follows that it is o-regular, too.
\qed
\end{proof}

We continue with repeating some crucial notions and results from \cite{He1}.
For a prefix  $\mbox{\rm X}\in\{\mbox{\rm D}, \mbox{\rm N}\}$, a language $L$ and a set $B$ of bounds, let
\bea
 \XT^L(B) & = & \bigcup\nolimits_{\beta\in B} \XT^L(\beta), \quad \mbox{where} \\
\XT^L(\beta) & = & 
\{L': \; L'\subseteq \X^* \mbox{ and $L'$ is accepted by an $\mbox{\rm X}$OTM $\frak{M}$  using $L$ as oracle } \\
 & & \hspace*{2.45cm}\mbox{and working with a related time complexity $t_{\frak{M}^L}\in \Ord(f)$} \}.
\eea
By an $\mbox{\rm X}$OTM, we mean a deterministic and nondeterministic, respectively, {\it oracle Turing machine}\/.
For further details, as the definition of time complexity of oracle machines, the reader is referred to textbooks  like \cite{DK,Pa,Re}.

For a class $\cal{C}$ of languages, let
\bea
 \XT^{\cal{C}}(B) & = & \bigcup\nolimits_{L\in\cal{C}} \XT^L(B).
\eea
In particular, if $B$ is an o-regular set of bounds and $\mbox{\rm X}\in\{\mbox{\rm D}, \mbox{\rm N}\}$, then it holds
\[\XT^{\DTi(B)}(B) = \XT(B). \]

The classes of the {\it oracle hierarchies}\/ are inductively defined analogously to the polynomial-time hierarchy 
(obtained for $B=\Bpol$) 
 and the linear-time hierarchy (obtained for $B=\Blin$). We put
\bea 
\Sa{1}{B} & = & \NT(B),\\
\Sa{k+1}{B} & = & \NT^{\Sai{k}{B}}(B) \quad \mbox{and}\\
\Pa{k}{B} & = & \co\Sa{k}{B}\en = \en \{\X^*\setminus L:\: L\in \Sa{k}{B}\,\}, \quad \mbox{for $k\in\Nplus$}, \quad \mbox{and}\\
\OH(B) &=& \bigcup\nolimits_{k\in\Nplusi} \Sa{k}{B}.
\eea

For every regular set $B$, the following language
\bea
V & = & \{ \langle w, \code(\frak{M})^t\rangle: \;
   w\in\X^*, |w|\leq t\in\Nplus \mbox{ and} \\
& & \hspace*{2.95cm} \mbox{$\,\frak{M}$ is a 2-tape $\Nd$TM   that accepts $w$ in $\leq t$ steps}\}
\eea
is $\NT(B)$-complete with respect to the relation of {\it $B$-reducibility}\/, $\leq_B$, between languages. This is defined as the m-reducibility by means of word functions computable with time bounds from $B$. More precisely, for $L,L'\subseteq \X^*$ let
\bea
L\leq_B L' & \mbox{ iff } & \mbox{there is a word function $\fct{f}{\X^*}{\X^*}$ which is computable
                        in time $\beta$,}\\
         & & \mbox{for some $\beta\in B$, \en such that for all $w\in\X^*$:
               \en    $w\in L \mbox{ iff } f(w)\in L'$}.
 \eea

Using the relativized versions of $V$, defined as
\bea
V^L & = & \{ \langle w, \code(\frak{M})^t\rangle: \;
   w\in\X^*, |w|\leq t\in\Nplus \mbox{ and $\frak{M}$ is a 2-tape $\Nd$OTM }\\
  & & \hspace*{3.0cm}
   \mbox{such that $\frak{M}^L$ accepts $w$ in $\leq t$ steps}\},
\eea
for any $L\subseteq \X^*$, we obtain universal $\Sa{k}{B}$-complete languages $V_k$ by  putting
\bea
V_1&=& V, \quad\mbox{ and}\\
V_{k+1}&=& V^{V_k}  \quad\mbox{ for all $k\in\Nplus$.}
\eea
For example, it follows that
\bea
\Sa{k+1}{B} &=& \NT^{V_k}(B) \quad \mbox{ for all $k\in\Nplus$},\\
 \OH(B) & = & \Sa{k}{B} \quad \mbox{ iff } \quad V_k\in \Pa{k}{B},
\eea
and a collapse of the oracle hierarchy of an o-regular set $B_1$ implies a collapse of the oracle hierarchy of each o-regular $B_2$ satisfying $B_1\leq B_2$. For details on these results and further ones, we refer to \cite{He1}.

It was not yet explicitly mentioned in \cite{He1} that
the oracle hierarchies can also be defined by means of ATMs ({\it alternating Turing machines}\/).  Let
\bea
 \Salt{k}{B} & = & \bigcup\nolimits_{\beta\in B} \Salt{k}{\beta}, \quad \mbox{where} \\
\Salt{k}{\beta} & = & 
\{L: \; L\subseteq \X^* \mbox{ and $L$ is accepted by a $\Sigma_k$-ATM $\frak{M}$  } \\
 & & \hspace*{5.4cm}\mbox{with a time complexity $t_{\frak{M}}\in \Ord(\beta)$} \}.
\eea
By a $\Sigma_k$-ATM, one usually means an ATM such that for each accepted input there is an accepting subtree $T$ whose root is existential and the number of quantifier changes along any computation path in $T$ is bounded by $k-1$. The $\Pi_k$-ATMs and the 
classes $\Palt{k}{B}$ are analogously defined, now with universal roots of the accepting subtrees. It easily follows that
$\Palt{k}{B}=\co\Salt{k}{B}$.
For further details, we again refer to standard textbooks.

\begin{proposition}
Let $B$ be an o-regular set of bounds and $k\in\Nplus$. Then
\[\Sa{k}{B}= \Salt{k}{B}.\]
\end{proposition} 

\begin{proof}
This is shown by induction on $k$. For $k=1$, the assertion holds due to
\[\Sa{1}{B}= \NT(B)= \Salt{1}{B}.\]
Now suppose that the assertion holds for some $k\in\Nplus$. Then the inclusion
\[\Sa{k+1}{B} \subseteq \Salt{k+1}{B}\]
follows by
\[ \Sa{k+1}{B} = \NT^{V_k}(B) \quad \mbox{ and } \quad V_k\in\Sa{k}{B} = \Salt{k}{B}.\]
Indeed, let $L\in\NT^{V_k}(\beta)=\NT^{\overline{V_k}}(\beta)$ be accepted by an NOTM $\frak{M}$ with oracle $V_k$ and with some time bound $\beta\in B$.
If $V_k$ is accepted by a $\Sigma_k$-ATM $\frak{M}_1$ within a time complexity $\beta_1\in B$,
its complement $\overline{V_k}$ can be accepted by a $\Pi_k$-ATM $\frak{M}_2$ with the same time bound $\beta_1$.
Moreover, due to the regularity of $B$, we can suppose that $\beta_1$ is time-constructible.

To accept $L$, a $\Sigma_{k+1}$-ATM $\frak{M}'$ first existentially guesses the sequence of oracle queries put by $\frak{M}$ along
a certain computation path as well as the answers according to the oracle set $V_k$. Then it confirms these answers by means of 
$\frak{M}_1$ and $\frak{M}_2$, respectively, and accepts the input if $\frak{M}$ would it accept on this path.
There is an $\Sigma_{k+1}$-ATM $\frak{M}'$ working in this way, and by means of the o-regularity of $B$, one can implement it
in such a way that the time complexity is bounded by some $\beta'\in B$.

To show the converse inclusion,
\[\Salt{k+1}{B} \subseteq \Sa{k+1}{B},\]
let $L$ be accepted by a $\Sigma_{k+1}$-ATM $\frak{M}$ with a time complexity bounded by some $\beta\in B$.
By separating the first stage of existential guessing of $\frak{M}$, we get a representation
\[L=\{ w:\: \exists^{\beta(|w|)}u\, \langle w,u\rangle\in L_0\,\}\]
with a language $L_0\in \Palt{k}{B}$. The meaning of the {\it bounded existential quantifier}\/ is as usual: for any $t\in\N$,
\[  \exists^t u \: \ldots \: \ldots \quad \mbox{ means } \quad \exists u (\,|u|\leq t \mbox{ and } \ldots \: \ldots\:).\]
$ \langle w,u\rangle$ denotes the result of a standard pairing function of $\X^*\times \X^*$ into $\X^*$ which is computable in linear time and has projections (assigning the components $w$ and $u$ to any $ \langle w,u\rangle$) which are linear-time computable, too.

The above representation of $L$ immediately yields that $L\in\NT^{L_0}(B)\subseteq \Sa{k+1}{B}$, since
$\Palt{k}{B}=\co\Salt{k}{B}=\co\Sa{k}{B}=\Pa{k}{B}$ by the hypothesis of induction. 
\qed
\end{proof}

For the sake of completeness, we still mention a third useful characterization of the classes of the oracle hierarchies.
It uses a prefix of bounded quantifications applied to languages from $\DT(B)$. The related details and proofs can be found in
\cite{Ko}. As usual, a {\it bounded generalization}\/ 
\[  \forall^t u \:\ldots \: \ldots \quad \mbox{ means } \quad \forall u (\,|u|\leq t \Longrightarrow \ldots \: \ldots\:).\]

\begin{proposition}[Kosel]
Let $B$ be an o-regular set of bounds, $k\in \Nplus$ and $L\subseteq \X^*$. Then \\[0.5ex]
$L\in \Sa{k}{B}$ \en iff \en there are $\:\beta_1, \ldots ,\beta_k\in B$ and $L'\in \DT(B)$ such that
\[ L = \{ w\in\X^*:\: \exists^{\beta_1(|w|)}u_1\, \forall^{\beta_2(|w|)}u_2\: \ldots \:Q^{\beta_k(|w|)}u_k\:
 \langle w,u_1,u_2, \,\ldots\ , u_k\rangle \in L'\,\},\]
where $Q\in\{\exists,\forall\}$ such that the prefix of quantifiers becomes alternating.
\qed
\end{proposition}

The code $\langle w,u_1,u_2, \,\ldots\ , u_k\rangle$  of a $(k+1)$-tuple $( w,u_1,u_2, \,\ldots\ , u_k )$ 
is straightforwardly defined by induction on $k$ and using the above pairing function $\langle \,\cdot\:,\,\cdot\,\rangle$.

In combination with facts of classical complexity theory, our results yield a separation of determinism from nondeterminism for 
several regular sets of bounds.

\begin{proposition}
Let $B$ be an o-regular set of bounds such that $\beta(n)\leqo n\cdot \log^*(n)$ for all $\beta\in B$ and there is a superlinear 
bound $\beta_0\in B$. Then it holds
\[\DT(B) \not= \NT(B).\]
\end{proposition}

\begin{proof}
This is indirectly shown. By means of the sketched results from \cite{He1}, from the assumption $\DT(B)=\NT(B)$  it would follow that $V\in\DT(B)$ and  $\OH(B)=\DT(B)\subseteq \DT( n\cdot \log^*(n))$.
By a classical result proved in \cite{PPST}, we have
\[  \DT( n\cdot \log^*(n)) \subseteq \Salt{4}{\Blin} = \Salt{4}{n}.\]
Hence our assumption would yield 
\[ \Salt{4}{\beta_0} \subseteq \OH(B) \subseteq \Salt{4}{n}\]
for any superlinear bound $\beta_0\in B$. But this is a contradiction to a key result from \cite{PR} which says that
$\, \Palt{k}{\beta_2}\setminus \Salt{k}{\beta_1} \not= \emptyset\,$
whenever $\beta_1\leo \beta_2$.
\qed
\end{proof}

\begin{corollary}
Let $\beta(n)=n\cdot \alpha(n)$ be a bound of Type 1 with an unbounded function $\alpha$ such that 
$\alpha(n)\llpow \log^*(n)$.
Then $\It(\beta)$ is an o-regular set and it holds 
\[\DT(\It(\beta))\not=\NT(\It(\beta)).\]
\end{corollary}

\begin{proof}
The o-regularity of $\It(\beta)$ holds by Proposition 6.
Since $\alpha$ is unbounded, $\beta$ is a superlinear function.
For $\beta'\in\It(\beta)$, by Lemma 13 there is an exponent $k\in\N$ such that $\beta'(n)\leqa n\cdot \alpha^k(n)$.
From the supposition of the corollary, it follows $\beta'(n)\leqo n\cdot \log^*(n)$, and Proposition 9 yields the 
claimed separation.
\qed
\end{proof}

In combination with the remarks at the end of Section 5, Corollary 5 shows that the separation of determinism from nondeterminism stated in Proposition 9 applies to a variety of o-regular sets of bounds below $n\cdot \log^*(n)$.

We would like to emphasize a typical advantage of regular sets of bounds on which the separation from Proposition 9
is based. As it was discussed in \cite{Gu}, it is open whether the equality $\DT(\beta) = \NT(\beta)$ implies
a collapse of the related alternating hierarchy, i.e., $\DT(\beta) = \Salt{k}{\beta}$ for all $k\in\Nplus$.
For o-regular sets $B$ of bounds, however, the equalities $\DT(B) = \NT(B)$ or $\Sa{k}{B}=\Sa{k+1}{B}$, for any $k\in\Nplus$,
cause a collapse of the whole oracle hierarchy (and, equivalently, of the alternating hierarchy) to $\DT(B)$ or to
its $k$-th level $\Sa{k}{B}=\Salt{k}{B}$.

Of course, if $\DT(B) \not= \NT(B)$, then there is a (time-constructible) $\beta\in B$ such that $\DT(\beta) \not= \NT(\beta)$.
By means of a suitable generalization of a transfer result, which was shown in \cite{Ka} for polynomial bounds, one could conclude 
$\DT(\beta) \not= \NT(\beta)$ for a variety of bounds $\beta(n)$ below $n\cdot\log^*(n)$. It is left to the interested reader to
elaborate these consequences of Proposition 9.

\section{Conclusion}

It was a leading idea of the project reported by the present paper to search for a universe of time bounds which, on the one hand, 
is large enough to obtain the main complexity classes, but which, on the other hand, is simple enough to allow some new insights.
The results summarized by Corollary 3 demonstrate an attempt in this respect. In Theorem 3, we used them to show  the richness of the system of (confinality classes of) regular sets of bounds. By Proposition 6, even the system of o-regular sets 
owns this richness. So the separation result of Proposition 9 earns some attention.

The order-theoretic view of Section 7 can now be combined with the results on universal complete languages obtained in 
\cite{He1} and reported in the previous section. For the sake of simplicity, it is useful to suppose that the underlying
universe of time bounds yields a linear order of confinality classes of regular sets.
\\[1ex]
{\bf Supposition 1.}\en
\it
 Let $U^*$ be a universe of time bounds such that the system of confinality classes of regular subsets
of $U^*$ is linearly ordered by $\leq$ canonically transferred from 
 $\Reg(U^*)$ to $\ov{\Reg(U^*)}$,
see the details at the beginning of Section 7.
This means that $\,(\ov{\Reg(U^*)}, \leq)\,$ is a linearly ordered set. 
Suppose, moreover, that there are regular sets $B_0, B^0\subseteq U^*$ such that 
\[ \DT(B_0) \not= \NT(B_0) \quad \mbox{ and } \quad 
   \DT(B^0) = \NT(B^0). \]
\rm

With the universal $\NT(B)$-complete language $V$, see \cite{He1} and the previous section, we have for every regular set $B$:
\[\DT(B)=\NT(B) \quad \mbox{ iff } \quad V\in\DT(B).\]
Thus, $V\not\in \DT(B_0)$ and $V\in\DT(B^0)$.
Moreover, if $B_1\leq B_2$ for any regular sets $B_1$ and $B_2$, then 
$\,\DT(B_1)=\NT(B_1)$, i.e.,  $V\in \DT(B_1)$, implies that
  $\, \DT(B_2)=\NT(B_2)$, i.e., $ V\in \DT(B_2)$.

So the determinism versus nondeterminism problem causes a cut $(C_-,C_+)$ of the linearly ordered set 
$\,(\ov{\Reg(U^*)}, \leq)$, where
\bea
C_- & = & \{[B]\in\ov{\Reg(U^*)} :\: V\not\in \DT(B)\,\} \quad \mbox{ and}\\
C_+ & = & \{[B]\in\ov{\Reg(U^*)}:\: V\in \DT(B)\,\}.
\eea
For example, we have $[B_0]\in C_-$ and $[B^0]\in C_+$.

Since for the union of all representatives of classes from $C_-$,
\[ B_- = \bigcup \{ B\in\Reg(U^*)    :\: V\not\in \DT(B)\},\]
it holds that $ V\not\in \DT(B_-)$, we have $[B_-]\in C_-$. Thus, $[B_-]$ is the greatest element of $C_-$
with respect to the relation $\leq$ between confinality classes of regular sets.
The representative $B_-$ even  is the largest (i.e., with respect to inclusion greatest) of all regular subsets $B\subseteq U^*$ satisfying
$V\not\in \DT(B)$, i.e., $\DT(B) \not=\NT(B)$.
\\[1ex]
{\bf Supposition 2.}\en
\it
Let $U_0^*$ be a nonempty set of time-constructible bounds which is linearly ordered by $\leqi$ such that
$2n\leqa \beta(n)$ for all $\beta\in U_0^*$ and $U^*=\It(U_0^*)$ is tame.
\vspace{1ex}

\rm
Then, by  Lemma 7,  $(\,\ov{\Reg(U^*)},\leq\,)$ is a linearly ordered set.
Moreover, $U^*$ is it-closed and it-complete, and it consists of time-constructible bounds only. By Corollary 3, there is a universe 
$U_0^*$ of this kind which satisfies
\[\Blin \,\cup\, \Blogs \,\cup\, \bigcup\nolimits_{m\in\Nplusi} \! \Bqmlin{m} \, \cup \,\Bpol\, \cup \,\Bqpol \; \emi \; U_0^*.\]
We can even ensure that 
\[U_{00}\, =\, \Blin \,\cup\, \Blogs \,\cup\ \bigcup\nolimits_{m\in\Nplusi} \! \Bqmlin{m} \, \cup \,\Bpol\,\cup \,\Bqpol \,\cup\,\Bhex
 \; \emi \; U_0^*. \]
Now Lemmas 19-23 as well as Proposition 4 apply to $U=U^*$, for all such sets $U^*=\It(U_0^*)$. Moreover,
if $B_0\equiv\Blin$ and $B_0\subseteq U^*$, then $\DT(B_0) \not= \NT(B_0)$, and 
if $B^0\equiv\Bhex$ and $B^0\subseteq U^*$, then $\DT(B^0) = \NT(B^0)$.
Thus, Supposition 1 is completely fulfilled by any such universe $U^*$.

The determinism versus nondeterminism problem even defines a cut $(B_{0,-},B_{0,+})$ of the linearly ordered set
$(U_0^*,\leqi)$, where 
\bea
B_{0,-} & = & \{\beta\in U_0^*: \: \It(\beta)\in C_-\} \quad \mbox{ and }\\
B_{0,+} & = & \{\beta\in U_0^*: \: \It(\beta)\in C_+\}. 
\eea
By Proposition 1 or Corollary 1, $B_{0,-}$ and $B_{0,+}$ are regular.
It is easily seen that $\,B_{0,-}=B_-$, the largest regular subset of $U^*$ for which determinism is properly weaker than 
nondeterminism.

Now $U_0^*$ is furthermore supposed to consist of a lower segment $U_0$ according to Corollary 3 and an upper segment $\widehat{U_0}$
containing some bound $\widehat{\beta}$ with $\It(\widehat{\beta})\equiv \Bhex$.
This means that $(U_0,\widehat{U_0})$ is a cut of the linearly ordered set $(U_0^*,\leqi)$.

It might be interesting to discuss the consequences of certain assumptions concerning determinism versus nondeterminism
with respect to the related cut of  $(U_0^*,\leqi)$.
For example, if we assume that $\DT(\Bqpol)=\NT(\Bqpol)$ or even that $\DT(\Bpol)=\NT(\Bpol)$, i.e., $\P=\NP$,
then the related regular set $B_-$ would be a proper subset of  $U_0$, and one could ask for its order-theoretic properties
within $(U_0,\leqi)$, maybe hoping to obtain a contradiction.

For example, it would be interesting to know whether the cut $(\,B_-\,,\, U_0\setminus B_-\,)$ of $(U_0,\leqi)$ could be a jump of
$(\Cuts(U_0),\incl)$.
By the way, Proposition 9 shows that $\It(n\cdot \log^*(n))\leq B_-$.
A related question asks to which points of the Cantor discontinuum $(\D,\leqx)$ this cut could correspond, with respect to
an order-isomorphism according to Proposition 5.
\vspace{0.5ex}

Quite a similar discussion could be elaborated with respect to the question of a collapse of the oracle hierarchies defined by
o-regular sets of bounds, instead of the determinism versus nondeterminism question. Now the universal $\Sa{k}{B}$-complete
languages $V_k$, for $k\in\Nplus$, should be employed.
Notice in this context that the set $\Bhex$ is o-regular, too.
For example, for any $k\in\Nplus$, the sets
\[B_{k,-} = \bigcup \{B:\: B\subseteq U^*, B \mbox{ is o-regular and } V_k\not\in\Pa{k}{B}\,\}\]
are the largest of all o-regular subsets $B$ of $U^*$ for which 
$\Sa{k}{B}\not=\Pa{k}{B}$, i.e., $\Sa{k}{B}\not= \OH(B)$.
Moreover, we have
\[  \ldots \;\leq \, B_{3,-} \,\leq \,  B_{2,-} \,\leq \,  B_{1,-}\,\leq \,  B_{0,-}\,,\]
where $B_{0,-}$ comes from the above discussion concerning determinism versus nondeterminism.
\vspace{1ex}

At this point, we want to close our final remarks which should emphasize the last sentence of the introduction.

\end{document}